\documentclass[sigconf, nonacm, pdfa]{acmart}
\usepackage[a-2b]{pdfx}
\hypersetup{pdfapart=2,pdfaconformance=b}
\usepackage{pvldb}% VLDB/PVLDB formatting overrides, see pvldb.sty for details.
\renewcommand\vldbdoi{10.14778/3849398.3849421}
\renewcommand\vldbpages{5286 - 5299}
\renewcommand\vldbavailabilityurl{https://github.com/lanfangping/CATune} 
\renewcommand\vldbvolume{19}
\renewcommand\vldbissue{13}
\renewcommand\vldbyear{2026}
\renewcommand\vldbpagestyle{empty}

\usepackage{multirow}
\usepackage{pdflscape} % rotate page to landscape
\usepackage{adjustbox} % rotate area
\usepackage{longtable} % for the long table that is broken into several pages
\usepackage{hyperref} % for hyperlinks
\usepackage[normalem]{ulem}  % 
\useunder{\uline}{\ul}{}
\usepackage{array}
\newcolumntype{L}[1]{>{\raggedright\arraybackslash}p{#1}}

\usepackage{booktabs}
\usepackage{amsmath}
\usepackage{makecell} 

\newcommand{\knob}[1]{\texttt{#1}}

\newcommand{\Sys}{\textsc{CATune}}

\makeatletter
\renewcommand{\subparagraph}{%
  \@startsection{subparagraph}{5}{\z@}{1ex \@plus 1ex \@minus .2ex}{-1em}{\normalfont\normalsize\itshape}}
\makeatother

\usepackage{amsthm} % For theorem-like environments
\usepackage{mdframed} % Include the mdframed package
\newtheoremstyle{italicstyle}
  {0em}   % Space above
  {0.5em}   % Space below
  {\itshape}  % Body font
  {}      % Indent amount (none)
  {\itshape} % Theorem head font
  {.}     % Punctuation after theorem head
  {\newline}  % Space after theorem head (space to insert after the head, before the text)
  {}  % Theorem head spec, use newline and set body indent to 2em

\theoremstyle{italicstyle}
\newtheorem{exampleinner}{Example}

\usepackage{setspace}  % For adjusting line spacing

\usepackage{hhline}

\usepackage{tikz}
\newcommand{\redcircled}[1]{%
    \tikz[baseline=(char.base)]{
        \node[shape=circle,fill=black,text=white,inner sep=1pt, yshift=0.5ex] (char) {#1};
    }%
}

\usepackage[ruled,vlined]{algorithm2e}
\usepackage{algorithmic}

\usepackage{graphicx}
\usepackage{subcaption}

\usepackage{enumitem}

\definecolor{reviewerone}{HTML}{005374}
\colorlet{revieweronebg}{reviewerone!2}

\definecolor{reviewertwo}{HTML}{006622}
\colorlet{reviewertwobg}{reviewertwo!2}

\definecolor{reviewerthree}{HTML}{6A3D9A}
\colorlet{reviewerthreebg}{reviewerthree!2}

\colorlet{colorcommentfg}{black}
\definecolor{colorcommentframe}{HTML}{005374}
\colorlet{colorcommentbg}{colorcommentframe!2}

\newcounter{reviewer}
\newcounter{revcomment}
\counterwithin*{revcomment}{reviewer}

\RequirePackage[breakable,skins,xparse]{tcolorbox}

\newlength{\borderthickness}
\newlength{\accentborderthickness}
\newcommand{\revcommentdefaultcolor}{colorcommentframe}
\newcommand{\revcommentdefaulttitle}[1]{Comment~#1}
\newcommand{\revresponselabel}{Response:}

\NewTColorBox[use counter*=revcomment]{revcomment}{ O{\revcommentdefaultcolor} O{} }{%
    adjusted title={%
        \if\relax\detokenize{#2}\relax
            \revcommentdefaulttitle{\therevcomment}%
        \else
            #2%
        \fi
    },%
    fonttitle={\bfseries\sffamily},%
    coltitle={#1},%
    coltext={colorcommentfg},%
    colback={#1!2},%
    colframe={#1},%
    colbacktitle={#1!2},%
    bookmark={%
        \if\relax\detokenize{#2}\relax
            \revcommentdefaulttitle{\therevcomment}%
        \else
            #2%
        \fi
    },%
    boxrule={\borderthickness},%
    leftrule={\accentborderthickness},%
    titlerule={\borderthickness},%
    enhanced jigsaw,%
    label={comment:\arabic{reviewer}.\arabic{revcomment}},%
}

\NewDocumentEnvironment{revresponse}{ O{} }{%
    \par\noindent\textsf{\textbf{\revresponselabel}}%
    \if\relax\detokenize{#1}\relax
        \enskip
    \else
        \enskip #1\par
    \fi
}{%
    \par
}

\begin{document}
\title{\Sys: Structural Constraint-Aware Bayesian Optimization for DBMS Configuration Tuning}

%%
%% The "author" command and its associated commands are used to define the authors and their affiliations.

\author{Fangping Lan}
\affiliation{%
  \institution{Temple University}
  \streetaddress{1801 N Broad St}
  \city{Philadelphia}
  \state{PA}
  \country{USA}
}
\email{fangping.lan@temple.edu}

\author{Qi Zhang}
\affiliation{%
  \institution{Temple University}
  \streetaddress{1801 N Broad St}
  \city{Philadelphia}
  \state{PA}
  \country{USA}
}
\email{qi.zhang@temple.edu}

\author{Eduard Dragut}
\affiliation{%
  \institution{Temple University}
  \streetaddress{1801 N Broad St}
  \city{Philadelphia}
  \state{PA}
  \country{USA}
}
\email{edragut@temple.edu}

\begin{abstract}
Modern DBMSs expose hundreds of configuration knobs, resulting in a high-dimensional and heterogeneous search space that makes automated tuning costly. 
Existing ML-based tuning systems typically treat the configuration domain as box-constrained and rely on workload feedback to implicitly capture inter-knob relationships. 
However, DBMS documentation specifies deterministic knob dependency constraints—particularly ordering constraints—that characterize structurally valid regions of the configuration space.
We present \Sys{}, a constraint-aware Bayesian optimization (BO) framework that models deterministic inter-knob ordering constraints as structural components of the search domain. 
Instead of learning feasibility boundaries through sampled violations, \Sys{} performs optimization within a constraint-consistent subspace. 
We develop a topology-aware sampling strategy that respects dependency structure during exploration and avoids the inefficiencies of post-hoc constraint handling. 
To enable automated constraint discovery, we further design a precision-first extraction pipeline that combines LLM-based parsing with reliability safeguards to mitigate hallucinated dependencies.
Experiments on PostgreSQL and MySQL using TPC-C and TPC-H workloads show that \Sys{} substantially improves both sample efficiency and final tuning quality across surrogate models and BO frameworks. 
Under default ranges, \Sys{} reaches the baseline optimum up to $12.5\times$ faster and improves throughput by up to 63.37\%. 
The improvements persist under knowledge-guided reduced ranges and alternative optimization implementations. 
These results demonstrate that explicitly modeling system-defined deterministic ordering constraints enhances optimization robustness and system stability.
\end{abstract}

\maketitle

\vldbtopmatter

\section{Introduction}

Modern Database Management Systems (DBMSs) expose hundreds of configuration knobs that govern memory allocation, concurrency control, logging behavior, and query execution strategies \cite{Pavlo2017SelfDrivingDM}. While this flexibility enables performance optimization across diverse workloads, it creates a high-dimensional and heterogeneous configuration space that is difficult to navigate \cite{10.14778/3450980.3450992}. For example, 
% PostgreSQL v14.9 provides 346 
PostgreSQL v13 provides 263
knobs spanning continuous, integer, and categorical domains. In cloud environments, where workloads and hardware configurations vary widely, manual tuning by database administrators (DBAs) becomes increasingly infeasible \cite{10.1145/3299869.3300085}.

To address this challenge, prior work has proposed automated and ML-based tuning methods. Existing approaches broadly fall into two categories:
(1) offline transfer-based methods that pre-train models on benchmark workloads and transfer knowledge to new deployments \cite{10.1145/3035918.3064029, Yogatama2014EfficientTL, 10.14778/3632093.3632114}, and
(2) online search-based methods that iteratively select configurations using search algorithms, reinforcement learning (RL) \cite{10.1145/3514221.3517882, 10.14778/3352063.3352129, 10.1145/3514221.3517843, 10.1145/3299869.3300085}, or Bayesian optimization (BO) techniques such as SMAC \cite{li2019acceleratingexperimentaldesignincorporating, 10.14778/1687627.1687767, 10.1145/3318464.3380591, 10.1145/3035918.3064029, 10.1145/3448016.3457291, 10.1145/3514221.3526176, 10.5555/3586589.3586643}.

Although these methods can eventually identify high-performing configurations, they incur significant tuning cost \cite{10.14778/3551793.3551844, 10.1145/3318464.3380591}. State-of-the-art systems often require \textit{hundreds of iterations} to stabilize performance \cite{10.14778/3538598.3538604}, and each iteration may execute expensive workloads. This cost arises from two challenges: the large number of knobs requiring exploration and the broad vendor-provided ranges associated with each knob. These ranges are intentionally conservative to ensure generality but substantially enlarge the search space and may include unstable or invalid regions \cite{10.14778/3476311.3476411, 10.14778/3450980.3450992}.

Meanwhile, human DBAs rarely rely on blind trial-and-error. Instead, they use domain knowledge documented in manuals and community discussions. Recent knowledge-aware tuners leverage Large Language Models (LLMs) to extract suggested values or reduce search regions \cite{10.1145/3514221.3517843, 10.1145/3514221.3520171, 10.14778/3659437.3659449, 10.1145/3733620.3733641, 10.1145/3626246.3654739}. However, these approaches primarily focus on identifying promising value ranges.

\begin{table}[]
\centering
\caption{Three examples of knob dependency constraints found in PostgreSQL V13 Documentation}
% \vspace{-1.2em}
\begin{tabular}{l|l}
\Xhline{1pt}
 \textbf{Knob} & \texttt{superuser\_reserved\_connections} \\
 \hline
 \textbf{Description} & \parbox[l][0.8cm][c]{6cm}{Determines the number of connection ... {\color{blue}The value must be less than \texttt{max\_connections}}...} \\
 \hline
 \textbf{Constraint} & \makecell[tl]{\texttt{superuser\_reserved\_connections} \\ $<$ \texttt{max\_connections}} \\
 \Xhline{1pt}
  \textbf{Knob} & \texttt{max\_parallel\_workers} \\
 \hline
 \textbf{Description} & \parbox[l][1.2cm][c]{6cm}{... note that {\color{blue} a setting for this value which is higher than \texttt{max\_worker\_processes} will have no effect}, since ...} \\
 \hline
 \textbf{Constraint} & \makecell[cl]{\texttt{max\_parallel\_workers} $\leq$ \\\texttt{max\_worker\_processes}} \\
 \Xhline{1pt}
  \textbf{Knob} & \texttt{max\_parallel\_workers\_per\_gather} \\
 \hline
 \textbf{Description} & \parbox[l][1.6cm][c]{6cm}{Sets the maximum number of workers ... {\color{blue}Parallel workers are taken from the pool of processes established by \texttt{max\_worker\_processes}, limited by \texttt{max\_parallel\_workers}} ...} \\
 \hline
 \textbf{Constraint} & \makecell[cl]{\texttt{max\_parallel\_workers\_per\_gather}\\
    $<$ \texttt{max\_worker\_processes},\\
    \texttt{max\_parallel\_workers\_per\_gather}\\
    $<$ \texttt{max\_parallel\_workers}} \\
 \Xhline{1pt}
\end{tabular}

\label{tab:constraint-examples}
% \vspace{-1.2em}
\end{table}

% \begin{table*}[t]
% \centering
% \caption{Three examples of knob dependency constraints found in PostgreSQL V13 Documentation}
% \vspace{-1em}
% \label{tab:constraint-examples}
% \begin{tabular}{p{3.9cm}|p{8cm}|p{4.8cm}}
% \Xhline{1pt}
%     \textbf{Knob} &
%     \textbf{Description} &
%     \textbf{Constraint} \\ \hline
    
%     \texttt{superuser\_reserved\_connec tions} &
%     Determines the number of connection “slots” that are reserved for connections by PostgreSQL superusers ... {\color{blue}The value must be less than \texttt{max\_connections}}...&
%     \texttt{superuser\_reserved\_connections} $<$ \texttt{max\_connections} \\  \hline 
    
%     \texttt{max\_parallel\_workers} & 
%     ... Also, note that {\color{blue} a setting for this value which is higher than \texttt{max\_worker\_processes} will have no effect}, since parallel workers are taken from the pool of worker processes established by that setting. & 
%     \texttt{max\_parallel\_workers} $\leq$ \texttt{max\_work er\_processes} \\  \hline

%     \texttt{max\_parallel\_workers\_per \_gather} &
%     Sets the maximum number of workers that can be started by a single Gather or Gather Merge node. {\color{blue}Parallel workers are taken from the pool of processes established by \texttt{max\_worker\_processes}, limited by \texttt{max\_parallel\_workers}}...
%     &
%     \makecell[tl]{\texttt{max\_parallel\_workers\_per\_gather}\\
%     $<$ \texttt{max\_worker\_processes},\\
%     \texttt{max\_parallel\_workers\_per\_gather}\\
%     $<$ \texttt{max\_parallel\_workers}} \\
% \Xhline{1pt}
% \end{tabular}
% \vspace{-1em}
% \end{table*}

An under-exploited form of knowledge lies in system-defined deterministic knob dependency constraints, i.e., relationships among knobs that define structurally invalid or suboptimal regions (e.g., ordering constraints such as $A \leq B$). These dependencies are documented in manuals and forums (Table~\ref{tab:constraint-examples}); to our knowledge, they are not explicitly encoded in existing ML-based tuning systems. As a result, current tuners treat the configuration space as box-constrained and implicitly learn such relationships through trial-and-error feedback, leading to wasted evaluations of invalid configurations. In contrast, we explicitly encode deterministic inter-knob ordering constraints as structural components of the search domain.

\textit{We hypothesize that explicitly encoding deterministic knob dependency constraints into the optimization process can substantially improve sample efficiency and robustness.} By filtering structurally invalid regions prior to evaluation, constraint-aware tuning reduces unnecessary trials and accelerates convergence.

We present \Sys{}, a system that encodes ordering constraints directly into the configuration domain and integrates them with BO-based optimizers. Ordering constraints naturally induce a structured dependency graph, which we model as a \emph{topology graph}. Unlike classical constrained BO methods \cite{Schonlau1998, 10.5555/3020751.3020778, 10.5555/3044805.3044997, 10.1007/s11222-014-9477-x, 10.1145/3448016.3457291}, which model constraints probabilistically and learn feasibility boundaries by sampling infeasible configurations, \Sys{} restricts optimization to a constraint-consistent subspace prior to search. Structural constraints are enforced in both exploitation (Section~\ref{sec:exploitation_phase}) and exploration (Section~\ref{sec:exploration_phase}).

% However, naive constraint enforcement via rejection-based handling can suffer from extremely low acceptance rates under multiple dependency constraints \cite{hutter2015smacmanual}. To address this, we design a topology-aware sampling strategy that respects the dependency graph during sample generation, ensuring feasibility by construction while maintaining high sampling efficiency (Section~\ref{sec:topo_sampling}).

However, post-hoc constraint-handling strategies enforce dependencies only after candidate generation, either by rejecting infeasible configurations~\cite{hutter2015smacmanual} or by repairing them into feasible ones. Under multiple dependency constraints, these approaches  incur substantial overhead from excessive rejected samples or repeated repair operations. To address this issue, we design a topology-aware sampling strategy that respects the dependency graph during sample generation, ensuring that every sampled configuration is feasible by construction while avoiding rejection and repair overhead.

Extensive experiments demonstrate that \Sys{} improves both final tuning quality and convergence speed across workloads, surrogate models, BO implementations, and search-space regimes (Section~\ref{sec:evaluation}). Compared to penalty-based strategies, proactive constraint filtering yields faster convergence by eliminating infeasible regions prior to optimization. Ablation studies confirm the effectiveness of topology-aware sampling (Section \ref{sec:different_strategies}). Beyond performance gains, we analyze the system-level impact of constraint enforcement (Section~\ref{sec:operational_impact}). Correct ordering constraints prevent DBMS startup failures and improve reliability, whereas hallucinated constraints can negatively affect optimization performance.

Although LLMs demonstrate strong capability in extracting tuning knowledge from textual sources \cite{wang2025autorulereasoningchainofthoughtextracted, zhang2024ruaglearnedruleaugmentedgenerationlarge}, they are not immune to hallucination \cite{51844, xu2025hallucinationinevitableinnatelimitation}. Our quality analysis in the technical report \cite{CATuneTechReport} shows that direct LLM prompting can produce unsupported or incorrectly oriented dependencies even when evidence spans are requested. Notably, two of these hallucinated constraints led to measurable degradation in tuning efficiency (Section \ref{sec:operational_impact}). 
% Our empirical analysis (Section~\ref{sec:operational_impact}) shows that among 30 extracted constraints, 6 were hallucinated. Notably, two of these hallucinated constraints led to measurable degradation in tuning efficiency. 
These findings highlight the risk of incorporating automatically generated constraints. 
To mitigate this risk, we design a precision-first constraint extraction pipeline that prioritizes correctness when leveraging LLM-based agents for automated dependency discovery.

\textbf{Contributions.} Our contributions are:

\begin{enumerate}[nosep]
    \item We propose \Sys{}, a DBMS tuning system that structurally encodes deterministic knob dependency constraints to restrict optimization to a constraint-consistent subspace.
    \item We compare three baseline constraint-handling strategies and show that correcting invalid configurations after sampling adds substantial overhead. 
    \item We introduce a topology-aware sampling strategy that avoids low acceptance rates inherent in rejection-based constraint handling.
    \item We design a precision-first LLM-assisted extraction pipeline that automatically discovers and refines knob dependency constraints while mitigating hallucination.
    \item We conduct extensive experiments demonstrating improved convergence speed and final performance across workloads, surrogate models, and BO frameworks.
    \item We provide a systematic analysis of constraint employment strategies, hallucinated constraints, and the system-level impact of constraint enforcement.
\end{enumerate}

\vspace{-10pt}
\section{Motivation} 
\label{sec:motivation}
%This section motivates our constraint-aware tuning framework. We demonstrate that existing ML-based tuners waste substantial budget due to invalid configurations (M1), despite the availability of rich structural knowledge in DBMS documentation (M2). We show that recent advances in LLMs enable scalable dependency extraction but introduce reliability risks (M3). Finally, we identify a fundamental gap: current tuning frameworks lack a systematic mechanism for enforcing structural knob dependencies during optimization (M4).

This section motivates constraint-aware tuning by showing that existing ML-based systems waste budget on invalid configurations despite documented structural dependencies, and that current frameworks lack a systematic mechanism for enforcing such dependencies during optimization.

\noindent\textbf{M1: Tuning budget waste significantly increases the cost of ML-based DBMS tuning.}
State-of-the-art ML-based tuning methods typically require hundreds of iterations to converge to a high-quality DBMS configuration, resulting in substantial tuning overhead. However, a non-trivial portion of this budget is effectively wasted. In practice, many sampled configurations are either invalid or uninformative.
 For example, in our preliminary experiments on PostgreSQL using the TPC-C workload with 46 knobs and documentation-specified search ranges, approximately 35\% of tuning iterations led to DBMS crashes. Additionally, numerous configurations violate knob dependency constraints (e.g., \knob{superuser\_reserved\_connections} must be strictly smaller than \knob{max\_connections}), producing executions that are rejected or semantically meaningless. 
Such wasted iterations not only increase operational cost but also slow model convergence. ML-based optimizers rely on observed performance signals to update surrogate models; invalid or low-quality configurations provide little useful information. As a result, substantially more iterations are required before the model forms reliable predictions, particularly in high-dimensional and constraint-heavy configuration spaces.

\begin{figure}[t]
    \centering
    \includegraphics[width=\linewidth]{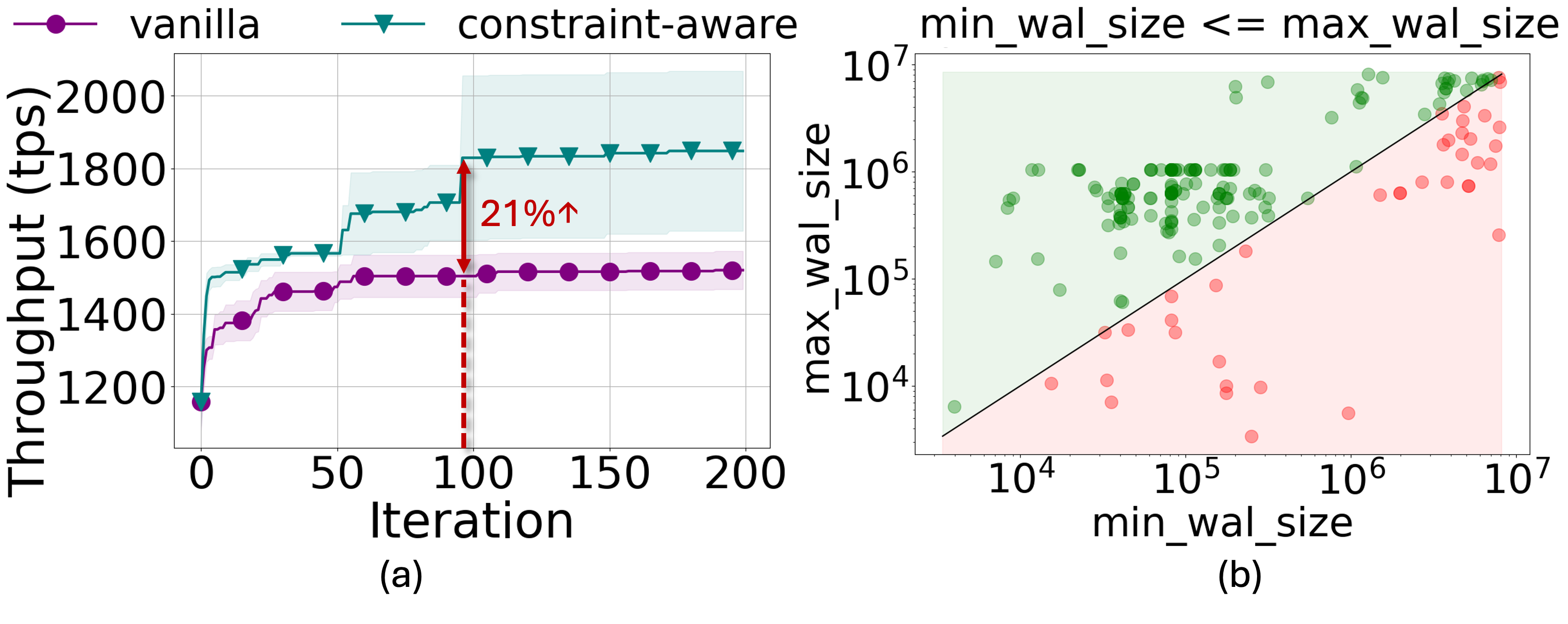}
    \caption{(a) The preliminary tuning performance improvement on the TPC-C workload using 10 knobs with/without constraint-aware tuning; (b) The feasible (green) and infeasible (red) samples/region during tuning based on the ordering constraint \texttt{min\_wal\_size} $<$  \texttt{max\_wal\_size}.}
    \label{fig:motivation-m2}
\end{figure}

\noindent\textbf{M2: Domain knowledge exists but remains underutilized.}
Extensive DBMS tuning knowledge has accumulated over decades and is documented in manuals, best practices, and expert guidelines. Prior work has incorporated portions of this knowledge into automated tuning frameworks, including suggested values \cite{10.1145/3514221.3517843, 10.1145/3514221.3520171}, range constraints, and special-value handling \cite{10.14778/3659437.3659449, 10.1145/3733620.3733641, 10.1145/3626246.3654739}. 
But, explicit \emph{inter-knob dependency relationships} remain largely unmodeled. These dependencies directly affect both configuration validity and achievable performance. Existing ML-based tuners implicitly attempt to learn such relationships from sampled observations, resulting in high sample complexity and slow convergence in high-dimensional domains. There is no guarantee that a learned surrogate model faithfully captures structural knob relationships.

As illustrated in Figure~\ref{fig:motivation-m2}, explicitly modeling knob dependencies significantly improves tuning efficiency. Incorporating constraints yields approximately 21\% higher throughput within the first 100 iterations, demonstrating faster convergence. Furthermore, enforcing constraints such as \knob{min\_wal\_size} $\leq$ \knob{max\_wal\_size} eliminates nearly 50\% of the nominal search region and reduces invalid samples by 19.5\% during tuning. These results highlight that dependency constraints meaningfully shrink the effective search space and reduce wasted trials. 
These observations motivate a systematic mechanism for modeling and enforcing structural inter-knob dependencies during optimization.

%Despite their importance, such dependency knowledge is often implicit, scattered across documentation, and primarily accessible to experienced DBAs. Existing automated tuning approaches either overlook these dependencies or encode them in limited, static forms, leaving substantial efficiency gains unrealized.

\noindent\textbf{M3: LLMs enable dependency extraction but introduce reliability challenges.} 
DBMS dependency constraints are expressed in diverse linguistic forms and are often scattered across documentation, making direct extraction challenging. Relation type and direction are easily confused \cite{pan-etal-2025-climateie,ZhangCPCLD24,zhang-etal-2025-dynclean}, especially when multiple knobs interact. Extraction errors are high-impact: hallucinated or incorrectly oriented constraints may prune feasible configurations or bias optimization.
To address this, we design a two-part solution (Section \ref{sec:constraint_extraction}): (i) a precision-first hybrid extraction core that combines schema-constrained LLM parsing with symbolic normalization and rule-based patterns, and (ii) an agentic reliability guardrail that applies self-reflection and LLM-as-judge verification before final conflict resolution. This design improves extraction trustworthiness while remaining lightweight and practical.

\noindent\textbf{M4: Structural constraints require dedicated integration into optimization.}
Modern optimization frameworks such as BO and RL are designed for black-box search over box-constrained domains. 
While effective for many tuning tasks, this abstraction does not directly account for deterministic inter-knob dependencies that induce structured feasible regions. 
Thus, structural constraint information must be incorporated through explicit modeling choices rather than relying solely on standard optimizer interfaces. 
This motivates a constraint-aware optimization framework that integrates structural dependency modeling directly into the search process.

%\textbf{M4: Existing optimizers lack systematic support for structural constraints.} Current optimization algorithms, including BO and RL, are designed primarily for black-box search and natively support only basic per-knob range constraints \cite{10.14778/3659437.3659449, 10.1145/3733620.3733641, 10.1145/3626246.3654739}. They do not directly support richer structural knowledge such as inter-knob dependencies without substantial modifications. Thus, documented constraint information is typically ignored or enforced through ad hoc heuristics. Such approaches are neither systematic nor scalable, particularly in high-dimensional configuration spaces where the search is inherently combinatorial. Without a dedicated constraint-aware mechanism, optimization remains inefficient and prone to exploring invalid or suboptimal configurations.

Together, these observations suggest that effective DBMS tuning requires enforcing documented knob dependencies directly within the optimization process. We formalize this constraint-aware tuning problem next.
\section{Problem Formulation}
\label{sec:problem_formulation}

\begin{figure*}
    \centering
    \includegraphics[width=0.85\linewidth]{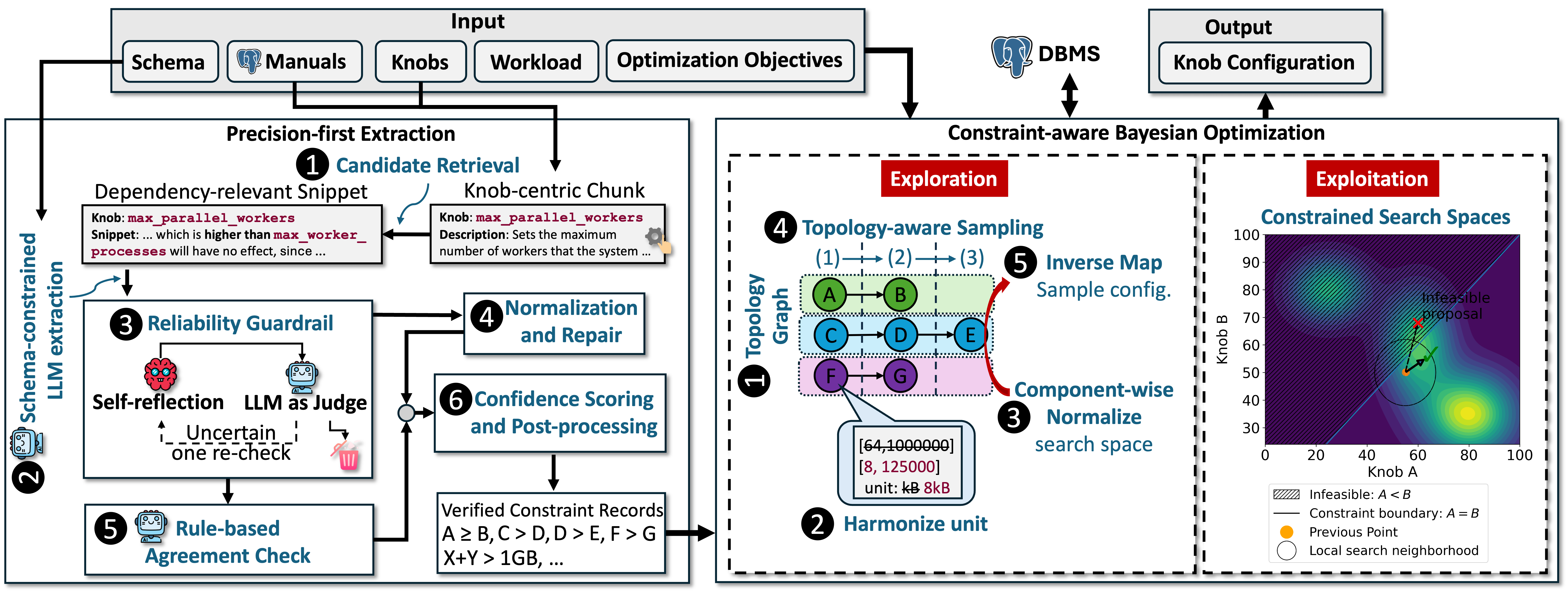}
    % \vspace{-1em}
    \caption{System overview of \Sys. 
\Sys{} consists of (1) a precision-first extraction pipeline that automatically derives ordering constraints from DBMS manuals, and (2) a constraint-aware Bayesian optimization engine that explicitly integrates these structural dependencies during both exploration and exploitation.}
    \label{fig:catune_design}
    % \vspace{-1.5em}
\end{figure*}

In this section, we formalize the constraint-aware DBMS knob configuration problem. We first revisit the classical formulation used in prior tuning work, and then introduce structural knob dependency constraints derived directly from DBMS manuals. We conclude by discussing the computational complexity of the resulting problem.

\noindent\textbf{Classical DBMS Knob Tuning.} 
%Modern DBMS expose tens to hundreds of configurable knobs that control memory allocation, concurrency, buffering, and execution behavior \cite{10.14778/3450980.3450992}. 
Let $\theta = (\theta_1, \ldots, \theta_n)$ denote the configuration vector, where each knob $\theta_i$ takes values in domain $\Theta_i$ (continuous, integer, or categorical). 
The nominal configuration space is
$
\boldsymbol{\Theta} = \Theta_1 \times \cdots \times \Theta_n.
$

Given a workload $W$, the DBMS performance (e.g., throughput or latency) is modeled as an unknown function
$
f : \boldsymbol{\Theta} \rightarrow \mathbb{R},
$
which can only be evaluated by executing the workload under configuration $\theta$.
The classical tuning problem is therefore defined as \cite{10.14778/3450980.3450992}:
\begin{equation}
\theta^* = 
\underset{\theta \in \boldsymbol{\Theta}}{\arg\max} 
\; f(\theta).
\end{equation}

Most existing tuning approaches treat $\boldsymbol{\Theta}$ as a box-bounded search space and perform black-box optimization over this space.

\noindent\textbf{Constraint-aware Tuning with Ordering Dependencies.}
DBMS manuals document deterministic ordering relationships between knobs. 
These dependencies specify relative constraints such as $\theta_i \le \theta_j$, 
ensuring valid configurations and preventing unsupported parameter combinations. 
Such relationships are known \emph{a priori} from documentation and can be verified without executing the DBMS (Table~\ref{tab:constraint-examples}). 
They restrict the set of admissible configurations before any performance evaluation is performed.
We model these dependencies as a set of pairwise ordering constraints, $\theta_i \le \theta_j, (i,j) \in \mathcal{C}$,
% \begin{equation}
% \theta_i \le \theta_j, \quad (i,j) \in \mathcal{C},
% \end{equation}
where $\mathcal{C}$ denotes the set of ordered knob pairs extracted from documentation.
The resulting feasible configuration space is $\boldsymbol{\Theta}_{\mathcal{C}} =
\left\{
\theta \in \boldsymbol{\Theta}
\;\middle|\;
\theta_i \le \theta_j \ \text{for all } (i,j) \in \mathcal{C}
\right\}.$
% \begin{equation}
% \boldsymbol{\Theta}_{\mathcal{C}} =
% \left\{
% \theta \in \boldsymbol{\Theta}
% \;\middle|\;
% \theta_i \le \theta_j \ \text{for all } (i,j) \in \mathcal{C}
% \right\}.
% \end{equation}

The constraint-aware tuning problem becomes:
\begin{equation}
\theta^* =
\underset{\theta \in \boldsymbol{\Theta}_{\mathcal{C}}}{\arg\max}
\; f(\theta).
\end{equation}

Unlike prior approaches that treat constraints as post-hoc performance conditions checked after execution (e.g., SLA violations), we enforce documented ordering dependencies directly in the search space. 
This eliminates structurally invalid configurations \emph{a priori} and fundamentally changes how candidate configurations are generated during optimization.

\subsection{Computational Complexity}

Even when restricted to ordering constraints, the DBMS knob tuning problem remains computationally hard.

\begin{lemma}
The constraint-aware DBMS knob tuning problem under a black-box objective is NP-hard.
\end{lemma}

\begin{proof}[Proof Sketch]
Consider a restricted instance in which all knobs are binary variables 
$\theta_i \in \{0,1\}$ and no ordering constraints are imposed.
Let the performance function be an arbitrary black-box mapping 
$f : \{0,1\}^n \rightarrow \mathbb{R}$. 
In the worst case, identifying the global optimum requires evaluating all 
$2^n$ possible configurations.
Thus, even without constraints, the tuning problem is combinatorial and NP-hard.
Since adding ordering constraints only restricts the feasible region without 
introducing exploitable structure in $f$, the general problem remains NP-hard.
\end{proof}

\noindent\textbf{Why the problem remains hard.}
The intrinsic difficulty stems from the objective function $f$, which is 
unknown, non-convex, and accessible only through 
costly workload executions. 
Ordering constraints reduce the feasible search region but do not impose 
structure on $f$ itself. 
Consequently, global optimization still requires exploration of a 
combinatorial configuration space.

\noindent\textbf{Feasibility checking.}
While optimization remains hard, checking whether the ordering constraints 
admit any feasible configuration is tractable. 
Pairwise ordering constraints of the form 
$\theta_i \le \theta_j$ can be represented as a system of difference 
constraints. 
Detecting inconsistency reduces to cycle detection in the induced directed graph 
or, equivalently, to a shortest-path feasibility check (e.g., via 
Bellman--Ford) \cite{math12162590}. 
Thus, verifying that the constraint set is consistent can be done in 
polynomial time, even though optimizing over the feasible region is NP-hard.
\section{\Sys\ Design}

\Sys~ is a constraint-aware DBMS tuning system that systematically incorporates documented structural knob dependencies into the optimization process. By enforcing these constraints during search, \Sys~ shrinks the effective configuration space and reduces wasted trials.

\noindent\textbf{Workflow.} Figure~\ref{fig:catune_design} illustrates the overall workflow. The user provides (i) the target DBMS (e.g., PostgreSQL), (ii) the corresponding manuals, (iii) a set of tuning knobs with their search ranges (which can be obtained via existing techniques~\cite{10.14778/3538598.3538604, 10.1145/3514221.3517843, 10.1145/3733620.3733641}), (iv) optimization objectives (e.g., throughput or latency), and (v) the target workload (e.g., TPC-H or TPC-C). 

\Sys~ first extracts knob dependency relationships from the manuals and refines them into machine-readable constraints (e.g., JSON format). It then performs constraint-aware BO to explore only feasible configurations and iteratively identify improved settings until the user-specified resource budget (e.g., maximum time or number of trials) is exhausted.

\sloppy
\noindent\textbf{Components.}
\Sys~ consists of two major components: (1) precision-first extraction (Section~\ref{sec:constraint_extraction}), and (2) constraint-aware Bayesian Optimization (Section~\ref{sec:constraint_aware_tuning}).

\noindent\textit{\underline{Precision-first Extraction.}}
To safely inject documentation-driven dependencies into tuning, \Sys~ adopts a precision-first extraction pipeline that favors reliable abstention over aggressive recall. 
Motivated by the reliability risks discussed in M3, the module couples (i) schema-constrained LLM parsing with symbolic normalization and high-precision rule templates, and (ii) an agentic guardrail that applies self-reflection and LLM-as-judge verification before conflict resolution. 
The output is a set of machine-readable constraint records (e.g., knob pair, relation type, optional activation condition, and evidence span) that can be deterministically validated and directly consumed by the downstream optimizer.

\noindent\textit{\underline{Constraint-aware BO.}}
\Sys~ modifies standard BO to enforce structural dependencies during both exploration and exploitation.

\textit{Exploration phase.}
Sampling under cross-knob constraints leads to low acceptance rates and wasted trials. 
To generate feasible configurations, \Sys~ performs five steps:
\redcircled{1} constructs a topology graph encoding knob dependency relationships;
% \redcircled{2} introduces control knobs to activate feature-dependent constraints when their conditions hold;
\redcircled{2} harmonizes heterogeneous units across knobs to enable consistent comparisons;
\redcircled{3} applies component-aware min–max normalization to transform knobs into a comparable search space;
\redcircled{4} performs topology-aware sampling in the normalized space, ensuring all generated configurations satisfy active constraints without rejection \cite{10.5555/3666122.3667026};
\redcircled{5} inverse-maps normalized samples back to their original value domains.

\textit{Exploitation phase.}
During exploitation, ordering and dependency constraints explicitly restrict the feasible region considered by the acquisition function. 
The surrogate model proposes new trials only within this constrained space, preventing evaluations in infeasible regions and improving sample efficiency.

\section{Constraint-Aware BO}
\label{sec:constraint_aware_tuning}

In this section, we describe how \Sys\ modifies standard BO to enforce structural knob dependencies during both exploration and exploitation. 
The key idea is to restrict the optimizer to valid configurations \emph{by construction}, thereby eliminating infeasible evaluations and improving sample efficiency.
Algorithm~\ref{alg:bo_db_constrained} outlines the overall procedure. 
We redefine the feasible configuration domain using the extracted dependency constraints (Line~4), so that both exploratory sampling (Line~14) and acquisition maximization (Line~16) operate only over valid configurations. 
Hence, the optimizer no longer wastes trials on structurally invalid settings.

\begin{algorithm}[t]
\caption{Constraint-Aware Bayesian Optimization for DBMS Tuning (differences to standard BO in {\color{blue}blue}).}
\label{alg:bo_db_constrained}
\LinesNumbered
\SetKwFunction{Train}{Train}
\SetKwFunction{Suggest}{SuggestNext}
\SetKwFunction{Build}{BuildTopoGraph}
\SetKwFunction{TopoSampler}{TopoSampler}
\SetKwFunction{Project}{ProjectToFeasible}
\SetKwFunction{Feasible}{IsFeasible}
\SetKwFunction{Norm}{Normalize}
\SetKwFunction{Inverse}{InverseMap}

\KwIn{
Database system $\mathsf{DB}$; nominal configuration space $\boldsymbol{\Theta}$;
{\color{blue}constraints $\mathcal{C}$};
performance metric $f(\cdot)$; objective type $\mathsf{obj} \in \{\text{throughput},\text{latency}\}$;
initial design size $n_0$; evaluation budget $T$.
}
\DontPrintSemicolon
\KwOut{Best configuration $\boldsymbol{\theta}^*$ \tcp*{Best among trials}}
\PrintSemicolon

Let $\mathcal{D} \leftarrow \emptyset$\tcp*{Tuning trajectories}
{\color{blue}Let $G\leftarrow \Build(\mathcal{C}, \boldsymbol{\theta})$}\tcp*{Build topology graph, see Algorithm \ref{alg:build-topology}}

{\color{blue} $\widehat{\boldsymbol{\Theta}} = \Norm(\boldsymbol{\Theta}, G)$}\tcp*{See Alg. \ref{alg:space-normalization}}

{\color{blue}Construct feasible space $\boldsymbol{\widehat{\Theta}}_{\mathcal{C}} \triangleq \{\boldsymbol{\hat{\theta}}\in\boldsymbol{\widehat{\Theta}}\mid c(\boldsymbol{\hat{\theta}})=\texttt{True},\ \forall c\in\mathcal{C}\}$}\;

\BlankLine
\textbf{(1) Feasible initial design.}\;
{\color{blue}$\{\boldsymbol{\hat{\theta}}_1,\ldots,\boldsymbol{\hat{\theta}}_{n_0}\} \leftarrow \TopoSampler(\boldsymbol{\widehat{\Theta}}, G, n_0)$}\tcp*{Alg. \ref{alg:topo-sampling-root-n}}
\For{$i=1$ \KwTo $n_0$}{
    Apply ${\color{blue}\boldsymbol{\hat{\theta}}_i}$ to $\mathsf{DB}$ and observe $y_i \leftarrow f({\color{blue}\boldsymbol{\hat{\theta}}_i})$\;
    $\mathcal{D} \leftarrow \mathcal{D} \cup \{({\color{blue}\boldsymbol{\hat{\theta}}_i}, y_i)\}$\;
}

\BlankLine
\textbf{(2) Iterative BO over the feasible space.}\;
\For{$t=n_0+1$ \KwTo $T$}{
    {\color{blue}$\mathcal{M} \leftarrow \Train(\mathcal{D})$}\tcp*{train surrogate on feasible data only}

    \eIf{$\mathrm{rand}(0,1) < p$}{
        {\color{blue}${\color{blue}\boldsymbol{\hat{\theta}}_t} \leftarrow \TopoSampler(\boldsymbol{{\color{blue}\widehat{\Theta}}}, G, 1)$}\tcp*{feasible exploration}
    }{
        {\color{blue}$\boldsymbol{\hat{\theta}}_t \leftarrow \Suggest(\mathcal{M}, \boldsymbol{\widehat{\Theta}}_{\mathcal{C}}, \mathcal{D})$}\tcp*{maximize acquisition within $\boldsymbol{\widehat{\Theta}}_{\mathcal{C}}$}
        % Optional if Suggest is unconstrained:
        % {\color{blue}\If{$\neg\Feasible(\boldsymbol{\theta}_t,\mathcal{C},\mathcal{A})$}{
        %     $\boldsymbol{\theta}_t \leftarrow \Project(\boldsymbol{\theta}_t, \boldsymbol{\Theta}_{\mathcal{C}}, G, \mathcal{A})$
        % }}
    }
    ${\color{blue}\boldsymbol{\theta_t} \leftarrow \Inverse(\boldsymbol{\hat{\theta}}_t)}$\tcp*{See Eq. \ref{eq:inv_norm} and \ref{eq:inv_harmonize}}
    
    Apply $\boldsymbol{\theta}_t$ to $\mathsf{DB}$ and observe $y_t \leftarrow f(\boldsymbol{\theta}_t)$\;
    $\mathcal{D} \leftarrow \mathcal{D} \cup \{(\boldsymbol{\theta}_t, y_t)\}$\;
}

\BlankLine
\eIf{$\mathsf{obj} = \mathrm{throughput}$}{
    $\boldsymbol{\theta^*} \leftarrow \arg\max_{(\boldsymbol{\theta},y)\in\mathcal{D}} y$\;
}{
    $\boldsymbol{\theta^*} \leftarrow \arg\min_{(\boldsymbol{\theta},y)\in\mathcal{D}} y$\;
}
\Return $\boldsymbol{\theta^*}$\;
\end{algorithm}

\subsection{Exploitation Phase}
\label{sec:exploitation_phase}

During exploitation, we restrict acquisition maximization to the feasible configuration space defined by the dependency constraints. 
Rather than modifying the surrogate model itself, we enforce constraint validation during candidate selection.
Specifically, when maximizing the acquisition function within a local neighborhood $\mathcal{N}(\boldsymbol{\theta}_t)$, candidate configurations that violate any dependency constraint are discarded. 
Thus, acquisition maximization is effectively performed over 
$
\mathcal{N}(\boldsymbol{\theta}_t) \cap \boldsymbol{\Theta}_{\mathcal{C}},
$
where $\boldsymbol{\Theta}_{\mathcal{C}}$ denotes the feasible space restricted by dependency constraints.

Figure~\ref{fig:catune_design} (Exploitation) illustrates this process under the ordering constraint $A \ge B$. 
While the unconstrained acquisition direction may point toward an infeasible region ($A < B$), constraint enforcement restricts the step to the admissible portion of the neighborhood. 
This ensures that all exploitation steps remain valid without altering the surrogate model.

% \begin{figure*}[]
%     \centering
%     \begin{minipage}[t]{0.2\linewidth}
%         \centering
%         \includegraphics[width=\linewidth]{figures/topological_graph.png}
%         \caption{The topology graphs induced by ordering constraints. }
%         \label{fig:topological_graph}
%     \end{minipage}
%     \hfill
%     \begin{minipage}[t]{0.75\linewidth}
%         \centering
%         \includegraphics[width=\linewidth]{figures/workflow_example.png}
%         \caption{An example of the exploration phase. }
%         \label{fig:workflow_example}
%     \end{minipage} 
% \end{figure*}

% \begin{figure}
%     \centering
%     \includegraphics[width=0.4\linewidth]{figures/topological_graph.png}
%     \caption{The topology graphs induced by ordering constraints.}
%     \label{fig:topological_graph}
% \end{figure}

\subsection{Exploration Phase}
\label{sec:exploration_phase}

\begin{figure}
    \centering
        \includegraphics[width=.9\linewidth]{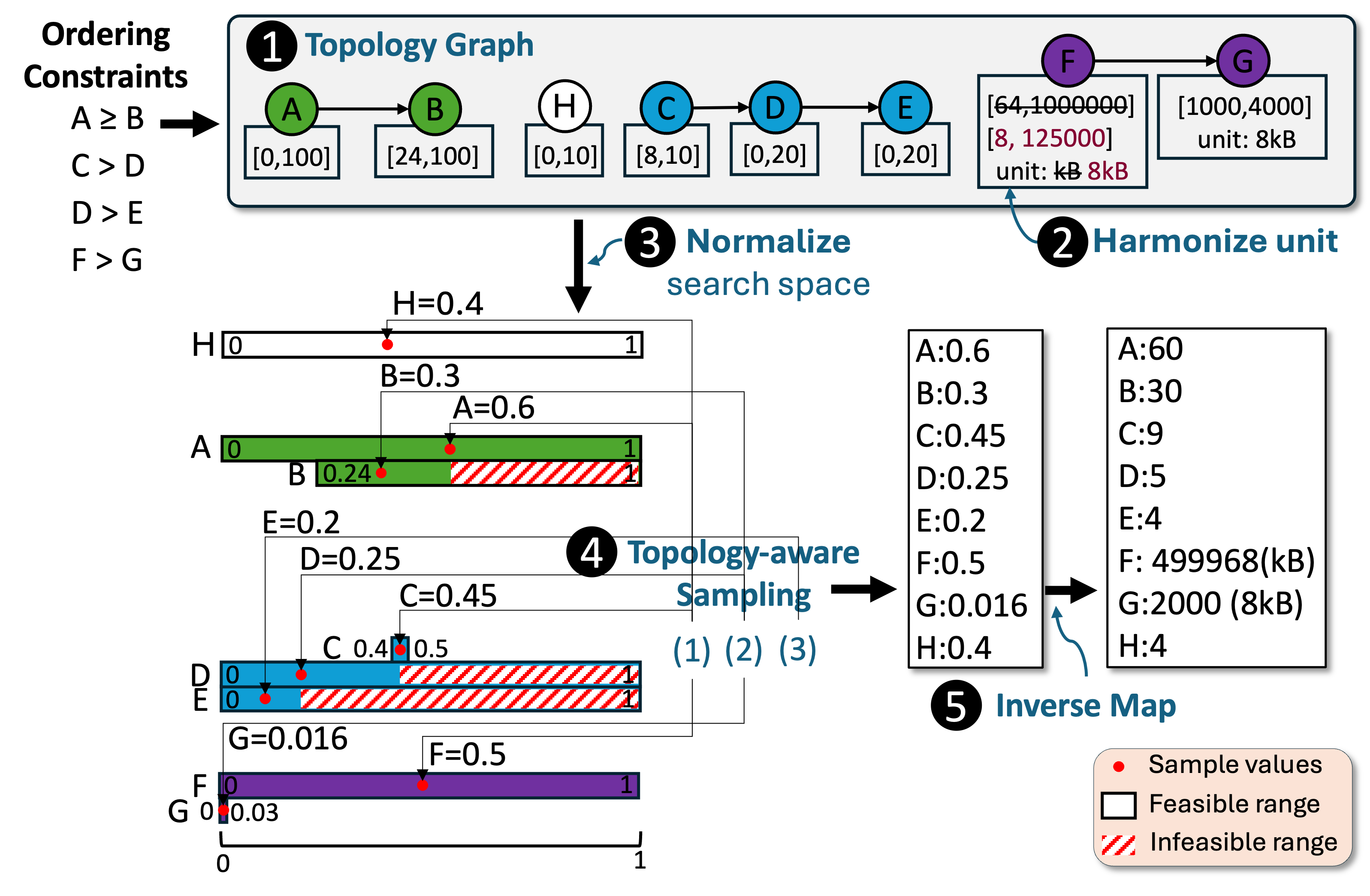}
        % \vspace{-1.8em}
        \caption{End-to-end illustration of the exploration phase. Given four ordering constraints, \Sys{} constructs a topology graph, harmonizes heterogeneous units, and applies component-wise normalization. Topology-aware sampling then follows the dependency order to generate feasible configurations with dynamically constrained ranges, after which sampled values are mapped back to their original knob units.}

        \label{fig:workflow_example}
        % \vspace{-1em}
\end{figure}
% End-to-end illustration of the exploration phase. Given four ordering constraints, including a chained dependency ($C > D > E$), \Sys{} constructs a topology graph in which the independent knob $H$ appears as an isolated node. Heterogeneous units, such as those of $F$ and $G$, are harmonized for consistent comparison. Each connected component is then normalized using component-wise min--max normalization while preserving the relative ordering structure. Topology-aware sampling generates a feasible configuration by following the dependency order and dynamically restricting feasible ranges. Finally, normalized values are mapped back to their original knob units through inverse transformation.

To enforce ordering constraints during the exploration phase, we design a five-step procedure that systematically ensures constraint consistency throughout the search process. Figure~\ref{fig:workflow_example} illustrates an end-to-end example of this workflow.

\subsubsection{Topology Graph}

Ordering constraints naturally induce a dependency structure among knobs. 
We represent this structure as a directed graph $G=(V,E)$, where each knob is a node and each ordering constraint defines a directed edge.
For a constraint $(a, \mathrm{op}, b)$, the edge direction reflects the sampling dependency implied by the operator. 
If $\mathrm{op} \in \{<, \le\}$ (i.e., $a$ must not exceed $b$), we add an edge $b \rightarrow a$, indicating that $b$ must be assigned before $a$. 
If $\mathrm{op} \in \{>, \ge\}$, we add an edge $a \rightarrow b$, meaning that $a$ must be assigned prior to $b$. 
Algorithm~\ref{alg:build-topology} describes this construction.

The resulting graph explicitly captures sampling dependencies among knobs and enables constraint-aware generation of configurations via topological traversal (Section~\ref{sec:topo_sampling}). 
After inserting all edges, we check for directed cycles. 
A cycle indicates mutually inconsistent ordering constraints and thus an invalid constraint set. 
Otherwise, the graph forms a directed acyclic graph (DAG) that defines a partial order over dependent knobs, while independent knobs remain isolated nodes, as illustrated in Figure~\ref{fig:workflow_example}.
This representation separates constraint reasoning from value sampling, allowing the optimizer to enforce dependencies without modifying the surrogate model.

\begin{algorithm}
\caption{Build a Topology Graph from Ordering Constraints}
\label{alg:build-topology}
\LinesNumbered
\KwIn{ordering constraint list $\mathcal{R}=\{(a,\mathrm{op},b)\}$, where $\mathrm{op}\in \{<, \le, >, \ge \}$; knob set $\mathcal{K}$.}
\KwOut{Directed graph $G=(V,E)$ with edge labels storing $\mathrm{op}$.}

$V \leftarrow \emptyset$, $E \leftarrow \emptyset$\;

\ForEach{$(a,\mathrm{op},b)\in \mathcal{R}$}{
    $V \leftarrow V \cup \{a,b\}$\;
    \uIf{$\mathrm{op}\in\{<,\le\}$}{
        $E \leftarrow E \cup \{(b \rightarrow a, \mathrm{op})\}$\;
    }
    \Else{
        $E \leftarrow E \cup \{(a \rightarrow b, \mathrm{op})\}$\;
    }
}

\tcp{Append independent knobs}
\ForEach{$k \in \mathcal{K}$}{
    \If{$k \notin V$}{
        $V \leftarrow V \cup \{k\}$\;
    }
}

\If{$G$ contains a directed cycle}{
    \Return \textsc{Fail}\;
}

\Return $G$\;

\end{algorithm}

\begin{algorithm}[t]
\caption{Search Space Normalization}
\label{alg:space-normalization}

% \SetKwFunction{Train}{Train}
% \SetKwFunction{Suggest}{SuggestNext}
% \SetKwFunction{Sample}{RandomSample}
% \SetKwFunction{Build}{BuildTopoGraph}
\SetKwFunction{Harmonize}{UnitHarmonize}
\SetKwFunction{Interval}{SharedInterval}
\SetKwFunction{Norm}{MinMaxNormalize}
% \SetKwFunction{TopoSampler}{TopoSampler}
\LinesNumbered
\KwIn{Knobs $\mathcal{K}=\{\theta_i\}$ with domains $[\ell_i,u_i]$ and unit $\mathsf{unit}_{i}$; topology graph $G$}
\KwOut{Normalized domains $\{\widehat{\Theta}_i\}$}

\ForEach{$\mathcal{G}_k$ in $G$}{
    \ForEach{$\theta_i \in \mathcal{G}_k$}{
        \textbf{1. Unit Harmonization}\; 
        $ \tilde{\theta}_i \leftarrow Harmonize(\theta_i)$ \tcp{Eq. \ref{eq:harmonize}} 

        \textbf{2. Component-wise min-max Normalization}\;
        $\ell^{(k)}, u^{(k)} = \Interval(\tilde{\theta}_i, \mathcal{G}_k)$ \tcp{Equation \ref{eq:component-interval}}
        $\widehat{\Theta}_i \leftarrow \Norm(\tilde{\theta}_i)$\tcp{Eq. \ref{eq:minmax_normalize}}
        % \quad $\widehat{\Theta}_i \leftarrow 
        %     \left[\frac{\ell_i-\ell^{(k)}}{u^{(k)}-\ell^{(k)}}, 
        %           \frac{u_i-\ell^{(k)}}{u^{(k)}-\ell^{(k)}}\right]$
        %  } 
    }
}

\Return $\{\widehat{\Theta}_i\}$\;
\end{algorithm}

\subsubsection{Search Space Normalization}

Ordering constraints require value comparison across knobs. 
However, DBMS knobs often use heterogeneous units. 
For example, \knob{shared\_buffers} may be specified in pages, while \knob{max\_wal\_size} is specified in MB, making raw comparisons invalid.
To ensure consistent constraint enforcement, we perform two preprocessing steps: (i) unit harmonization and (ii) component-wise normalization (Algorithm~\ref{alg:space-normalization}).

\textbf{Unit Harmonization.}
For each semantic knob type, such as memory or time, \Sys{} automatically extracts the knob unit from the documentation and converts all compatible knobs to a canonical unit.
For a knob $\theta_i$ with domain $\Theta_i=[\ell_i,u_i]$ and unit $\mathsf{unit}_i$, let $\mathsf{unit}^{(t)}$ denote the canonical unit for its type $t$. 
We apply a deterministic linear conversion:
\begin{equation}
\label{eq:harmonize}
\tilde{\theta}_i = \phi_i(\theta_i) = \alpha_i \theta_i, 
\qquad 
\tilde{\Theta}_i = [\alpha_i \ell_i, \alpha_i u_i],
\end{equation}
where %$\alpha_i$ converts $\mathsf{unit}_i$ to $\mathsf{unit}^{(t)}$.
 $\alpha_i$ is the deterministic conversion factor from $\mathsf{unit}_i$ to $\mathsf{unit}^{(t)}$. 

\textbf{Component-wise Normalization.}
%Since ordering constraints form connected components in the topology graph, we normalize knobs within each component consistently. 
Since ordering constraints connect knobs that must be numerically comparable (after unit harmonization), each connected component defines a group of knobs sharing a common comparison space.
For a component $\mathcal{G}_k$, we define the shared interval:
\begin{equation}
\label{eq:component-interval}
\ell^{(k)} = \min_{\theta_i \in \mathcal{G}_k} \tilde{\ell}_i, 
\qquad
u^{(k)} = \max_{\theta_i \in \mathcal{G}_k} \tilde{u}_i,
\end{equation}
and normalize each knob $\theta_i \in \mathcal{G}_k$ as
\begin{equation}
\label{eq:minmax_normalize}
\hat{\theta}_i = \frac{\tilde{\theta}_i - \ell^{(k)}}{u^{(k)} - \ell^{(k)}}.
\end{equation}
Different components are normalized independently.

\begin{example}
Suppose $\texttt{work\_mem} \leq \texttt{shared\_buffers}$ with domains $[4,512]$ MB and $[128,8192]$ MB, respectively.
Independent normalization maps both domains to $[0,1]$, distorting their relative scale.
Component-wise normalization instead uses the shared interval $[4,8192]$, preserving ordering consistency in the normalized space.
\end{example}

\textbf{Inverse Mapping.}
Before applying a configuration to the DBMS, we invert the preprocessing steps. 
For $\hat{\theta}_i \in [0,1]$ in component $\mathcal{G}_k$,
\begin{equation}
\label{eq:inv_norm}
\tilde{\theta}_i = 
\hat{\theta}_i \left(u^{(k)} - \ell^{(k)}\right) + \ell^{(k)},
\end{equation}
and the original knob value is recovered as
\begin{equation}
\label{eq:inv_harmonize}
\theta_i = \tilde{\theta}_i / \alpha_i.
\end{equation}

\begin{example}
Suppose \knob{work\_mem} belongs to a connected component with shared interval
$\ell^{(k)} = 4$ MB and $u^{(k)} = 8192$ MB.
Assume its harmonization factor is $\alpha_i = 1$ (already expressed in MB).
If the optimizer proposes a normalized value $\hat{\theta}_i = 0.5$,
the harmonized value is recovered via Equation~\eqref{eq:inv_norm}:
$\tilde{\theta}_i
=
0.5 \times (8192 - 4) + 4
=
4098 \text{ MB}$.
% \[
% \tilde{\theta}_i
% =
% 0.5 \times (8192 - 4) + 4
% =
% 4098 \text{ MB}.
% \]
Since $\alpha_i = 1$, the original knob value is $\theta_i = 4098$.

If instead the knob were originally specified in pages
(e.g., 1 page = 8 KB) and $\alpha_i$ converted pages to MB,
Equation~\eqref{eq:inv_harmonize} would restore the value to its native unit
before applying it to the DBMS.
This design allows \Sys{} to operate in a unit-consistent normalized space during optimization, while ensuring that deployed configurations remain expressed in their native DBMS units. 
By decoupling optimization from unit heterogeneity, we preserve semantic correctness without complicating constraint enforcement.
\end{example}

\begin{algorithm}[t]
\caption{ Topological Sampler}
\label{alg:topo-sampling-root-n}
\SetKwFunction{Harmonize}{UnitHarmonize}
\SetKwFunction{Norm}{MinMaxNormalize}
\LinesNumbered
\KwIn{
Configuration space $\mathcal{C}$ with $d$ knobs;
topology graph $G=(V,E)$;
% constraint activation rules $\mathcal{A}$;
number of samples $N$.
}
\KwOut{A set of valid configurations $\mathcal{S} = \{\mathbf{c}_1,\dots,\mathbf{c}_N\}$.}

Compute a topological order $\pi \leftarrow \mathrm{TopoSort}(G)$\\
% Let $\mathcal{R} \leftarrow \{v \in V \mid \mathrm{indegree}(v)=0\} \cup \{hp\in\mathcal{C} \mid hp\notin V\}$ \tcp*[f]{root knobs and independent knobs}\\
Let $\mathcal{R} \leftarrow \{v \in V \mid \mathrm{indegree}(v)=0\}$ \tcp*[f]{root knobs and independent knobs}\\
Initialize $\mathcal{S} \leftarrow \emptyset$ \tcp*[f]{sample set}

% \BlankLine
\tcp{(1) Sample roots}
Sample $\mathbf{U} \in [0,1]^{N \times |\mathcal{R}|}$ for $\mathcal{R}$ using LHS\;
\For{$i = 1 \dots N$}{
    Initiate an empty configuration $\mathbf{c}_{i}$\;
    \For{$j = 1 \dots |\mathcal{R}|$}{
        $\mathbf{c}_{i}[hp_{j}] \leftarrow \text{convert}(u_{ij}, hp_j)$ 
        % \eIf{$u_{ij}\in \mathbf{U}$ is continuous}{
        %     $x_{ij} \leftarrow \text{scale}(u_{ij}, hp_j)$\;
        % }{
        %     $x_{ij} \leftarrow \text{map}(u_{ij}, hp_j)$\;
        % }
    }
}
\tcp{(2) Sample remaining knobs in topological order}
\For{$i=1$ \KwTo $N$}{
    % Initialize partial assignment $\mathbf{v}\leftarrow \emptyset$
    
    \ForEach{$x \in \pi \setminus \mathcal{R}$}{
        % \If{$\neg \mathrm{Active}(x,\mathbf{c}_{i},\mathcal{A})$}{
        %     \textbf{continue}
        % }

        Initialize feasible space $I_x = \text{def\_space}(x) = [l_{x}, h_{x}]$\\
        \ForEach{incoming edge $(p \rightarrow x, \mathrm{op})\in E$}{
            $h_x \leftarrow \min(h_x, \mathbf{c}_{i}[p])$ 
            \tcp{Tighten $I_x$}
        }
        % $I_{x}= [l_{x}, h_x]$ \\
        Sample $\mathbf{c}_{i}[x]$ uniformly from $I_x$\\
        
    }

    % Construct configuration $\mathbf{c}_i$ from $\mathbf{v}$\\
    % Add $\mathbf{c}_i$ to $\mathcal{S}$
    $S \leftarrow S \cup \{\mathbf{c}_{i}\}$
}

\Return $\mathcal{S}$
\end{algorithm}
\subsubsection{Topology-aware Sampling}
\label{sec:topo_sampling}
To enable efficient exploration of the feasible region, we propose a topology-aware sampler that eliminates the inefficiencies of rejection-based constrained sampling.

\noindent\textbf{Motivation.} A straightforward way to enforce ordering constraints is rejection-based sampling: repeatedly draw a configuration from the full Cartesian product and discard it if any constraint is violated~\cite{hutter2015smacmanual}. We compare this and other post-hoc constraint-handling strategies in Section~\ref{sec:efficiency_topo_sampling}.
While simple, these post hoc approaches become impractical when the number of ordering constraints increases, especially when they are chained, and their search spaces are extremely imbalanced. 
In such cases, the feasible region may occupy only a small fraction of the full search space, and the probability of sampling a valid configuration drops multiplicatively as constraints accumulate. 
As a result, the optimizer may require an excessive number of trials to obtain a single feasible configuration.

% \textbf{Why Rejection-based Strategy Fails.}
% A straightforward way to enforce ordering constraints is a rejection-based strategy: repeatedly draw a configuration from the full Cartesian product and discard it if any constraint is violated \cite{hutter2015smacmanual}. 
% While simple, this approach becomes impractical when the number of ordering constraints increases, especially when they are chained, and their search spaces are extremely imbalanced. 
% In such cases, the feasible region may occupy only a small fraction of the full search space, and the probability of sampling a valid configuration drops multiplicatively as constraints accumulate. 
% As a result, the optimizer may require an excessive number of trials to obtain a single feasible configuration.

\noindent\textit{Case Study.}
Consider the constraints 
\texttt{superuser\_reserved\_ connections} (range: [0, 262143]) $<$ \texttt{max\_connections} (range: [24, 100]) and 
\knob{max\_wal\_senders} (range: [2, 8000]) $<$ \texttt{max\_connections}. 
Due to highly imbalanced value ranges, the probability of randomly satisfying each constraint can be extremely small. 
When combined, the joint acceptance probability may fall to around $10^{-6}$, implying millions of random trials per feasible sample. 
Such inefficiency severely degrades the optimization process.

\noindent\textbf{Topology-Aware Sampling.}
To avoid rejection, we generate configurations directly in a dependency-consistent order defined by the topology graph \cite{charpentier2022differentiabledagsampling, 9341040}.
We first identify root knobs (nodes with no incoming edges) and sample them using Latin Hypercube Sampling (LHS) to ensure good coverage of their domains \cite{10.5555/3488733.3488749, 10.1145/3318464.3380591, 10.5555/846219.847390, 10.1145/3448016.3457291}. 
We then traverse the graph in topological order. 
When sampling a knob, we dynamically restrict its feasible interval based on the already assigned values of its parent knobs. 
Thus, every sampled value automatically satisfies all ordering constraints.

This sequential, constraint-aware sampling eliminates rejection entirely: every generated configuration is feasible by construction. 
Moreover, sampling cost grows linearly with the number of knobs rather than being governed by the shrinking acceptance probability of the full Cartesian space.
%\noindent\textit{Example.}
Figure~\ref{fig:workflow_example} illustrates this process. 
Root nodes are sampled first. 
Subsequent knobs are sampled within dynamically restricted intervals determined by their parents, progressively pruning infeasible regions. 
After all knobs are assigned, inverse normalization restores values to their original DBMS units.
\section{Knob Constraint Extraction}
\label{sec:constraint_extraction}

\Sys{} requires reliable evidence of knob dependencies before they can be incorporated into the configuration search space. However, these dependencies are expressed in heterogeneous forms across DBMS documentation, and an incorrectly extracted relation may exclude valid configurations. We thus develop a precision-first pipeline that extracts structured dependencies from manuals while filtering unsupported and ambiguous relations.

\subsection{Extraction Scope and Inputs}
\label{sec:extraction_scope}

The extraction module takes three inputs: (1) a version-specific corpus of official DBMS documentation \cite{LanZD26}, (2) a target knob list, and (3) a predefined relation schema. The knob list can be obtained using existing knob selection techniques, such as SHAP~\cite{NIPS2017_8a20a862}, Lasso~\cite{10.1111/j.2517-6161.1996.tb02080.x}, LLM-assisted knob selection approaches~\cite{10.1145/3733620.3733641}, or manual selection. Knob selection itself is outside the scope of this paper. The relation schema defines the dependency types recognized by the extraction pipeline. This preparation is performed once per DBMS version.

Given these inputs, candidate retrieval, extraction, verification, normalization, and filtering are performed automatically. Thus, aside from providing the documentation corpus, relation schema, and target knob list, the extraction workflow is automated. Each extracted record contains \texttt{(knob1, relation, knob2, condition, context, evidence, confidence)}.

The current optimizer in \Sys{} enforces verified pairwise numerical ordering constraints, such as $A<B$ and $A\leq B$. The extraction pipeline also identifies activation conditions, resource bounds, multi-knob dependencies, and performance recommendations, but these relation types are retained rather than translated into topology constraints by the current optimizer.

Our design prioritizes precision because extraction errors have asymmetric downstream effects. A missed constraint leaves part of the search space unpruned, whereas a hallucinated or incorrectly oriented constraint may exclude valid and potentially high-performing configurations. Consequently, \Sys{} retains only evidence-supported relations and abstains when the documentation does not provide sufficient support. Its scope is intentionally limited to dependencies explicitly stated in the selected manual version; undocumented, workload-specific, or cross-version dependencies are outside the scope of the current system.

\subsection{Precision-First Extraction Pipeline}
\label{sec:constraint_workflow}

Figure~\ref{fig:catune_design} presents the extraction pipeline. It combines schema-constrained LLM extraction with evidence verification, canonical normalization, and deterministic rules for recurring documentation expressions.

\noindent\textbf{(1) Candidate retrieval.}
We parse the documentation into knob-centric entries and identify dependency-relevant paragraphs using knob co-mentions and trigger expressions, such as ``must be less than,'' ``at least,'' ``ignored unless,'' and ``has no effect''~\cite{postgresql13-resource-config,postgresql13-wal-config,mysql80-system-variables}. The trigger expressions are derived from the closed relation schema and fixed before extraction. They are used only to retrieve candidate passages and do not establish constraints by themselves. Each selected paragraph is expanded with its local context and divided into bounded-size chunks. This step reduces irrelevant input while preserving the context needed to interpret the relation.

\noindent\textbf{(2) Schema-constrained extraction.}
For each candidate chunk, the LLM receives the primary knob, the knobs mentioned in the local context, and a closed set of relation labels \cite{zhang-etal-2026-scaling}. It must return structured JSON records and provide a supporting evidence span for every extracted relation. If the chunk contains no explicit dependency, the model is instructed to return an empty result. Requiring source-grounded evidence makes each extraction traceable and supports deterministic verification~\cite{li2024simple,huang-etal-2024-learning,pan-etal-2025-taxonomy}.

\noindent\textbf{(3) Evidence verification.}
Each preliminary record passes through two verification stages. First, Self-Reflection may revise or remove unsupported records and correct relation, direction, or condition errors~\cite{ji-etal-2023-towards,shinn2023reflexionlanguageagentsverbal}. Second, a separate LLM-as-Judge call, using a verification-specific prompt and the same model backend by default, scores the revised record on three \emph{label-agnostic} axes: whether the snippet states the dependency, whether the operands are assigned to the correct roles, and whether any activation condition is faithful~\cite{10.5555/3666122.3668142}. 
A record is rejected if any axis fails.
The judge does not choose the ``most accurate'' relation label because documentation can express the same dependency in multiple ways. Instead, it verifies evidence support, operand direction, and conditions against the cited snippet. Records with uncertain support trigger one additional judge call and are retained only when both calls agree on the same canonical relation; otherwise, the pipeline abstains.

\noindent\textbf{(4) Normalization and semantic repair.}
We normalize knob names, relation labels, activation conditions, and relation directions. Semantically equivalent descriptions are mapped to the same representation. For example, both ``$A$ must not exceed $B$'' and ``$B$ must be at least $A$'' are normalized to the canonical relation $A\leq B$. This canonicalization prevents semantically equivalent statements from producing opposite edge orientations in the topology graph.

\noindent\textbf{(5) Deterministic rules.}
Generic comparison templates handle recurring documentation expressions, such as mapping ``must be less than $X$'' to an ordering relation. This rule branch adds evidence for explicit patterns, while each extracted record still undergoes evidence verification and canonical normalization \cite{zhang2019invest,ZhangHVD18}.

\noindent\textbf{(6) Scoring and post-processing.}
A record is retained when it passes the judge's support, direction, and condition checks (Stage~3). Reflection consistency, judge support, and rule agreement are combined into a reliability score used to rank duplicate or conflicting records. We then deduplicate identical tuples and resolve conflicting relations for the same directed knob pair by retaining the best-supported record. The retained constraints are then passed to topology construction, where inconsistent directed cycles are detected as described in Algorithm~\ref{alg:build-topology}.

\begin{example}
Consider the documentation statement that a value of \texttt{max\_parallel\_workers} higher than \texttt{max\_worker\_processes} has no effect~\cite{postgresql13-resource-config}. For readability, let $A$ and $B$ denote these two knobs, respectively. Candidate retrieval selects this statement based on the knob co-mention and the phrase ``higher than.'' The LLM branch extracts the relation and its evidence span, which normalization canonicalizes to $A \leq B$.
The verification stage then confirms that the snippet states the dependency, that $A$ and $B$ occupy the correct roles, and that no spurious condition was added, so the record is retained. A record that failed any of these checks (e.g., one asserting the reverse ordering) would be rejected.
\end{example}

\begin{figure*}[t]
    \centering
    \begin{minipage}[t]{0.48\linewidth}
        \centering
        \includegraphics[width=\linewidth]{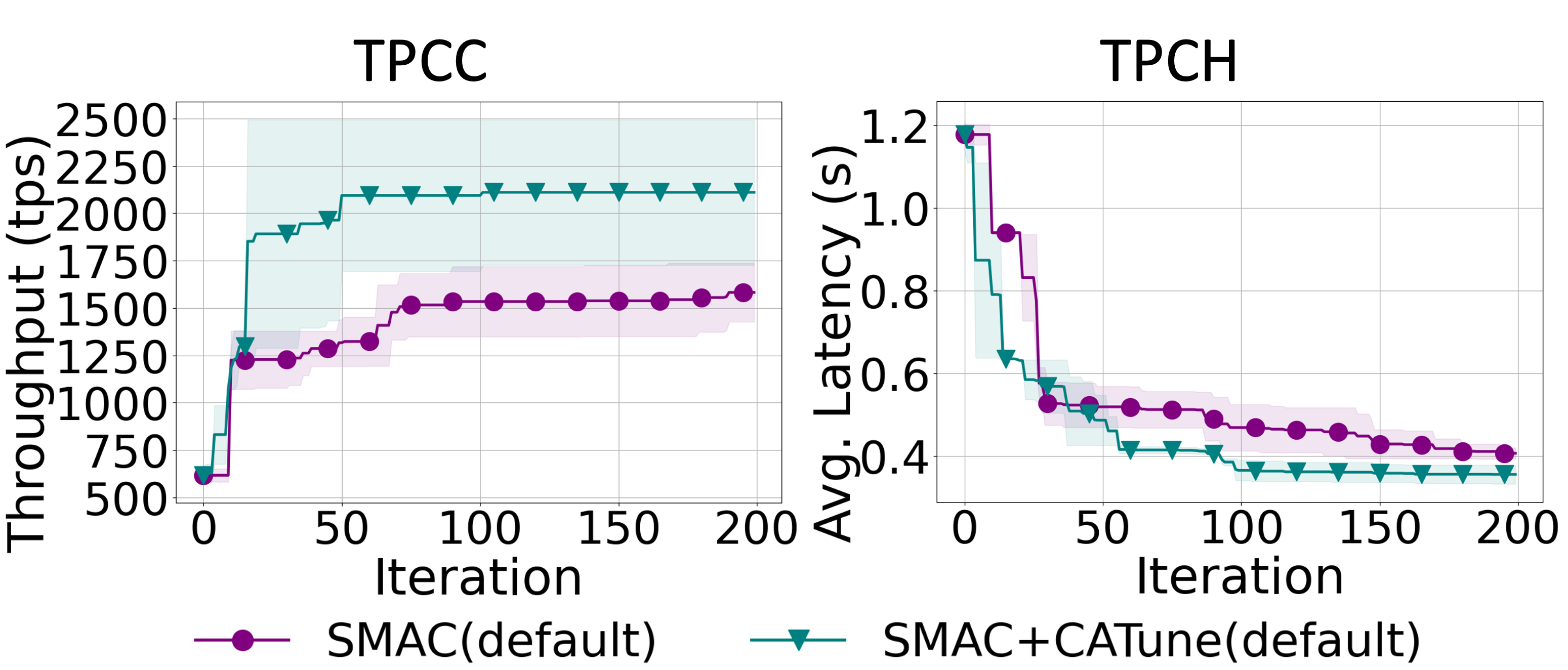}
        % \vspace{-2em}
        \caption{Best performance achieved by \Sys{} under manual default ranges on PostgreSQL.}
        \label{fig:efficiency_performance_default}
    \end{minipage}
    \hfill
    \begin{minipage}[t]{0.48\linewidth}
        \centering
        \includegraphics[width=\linewidth]{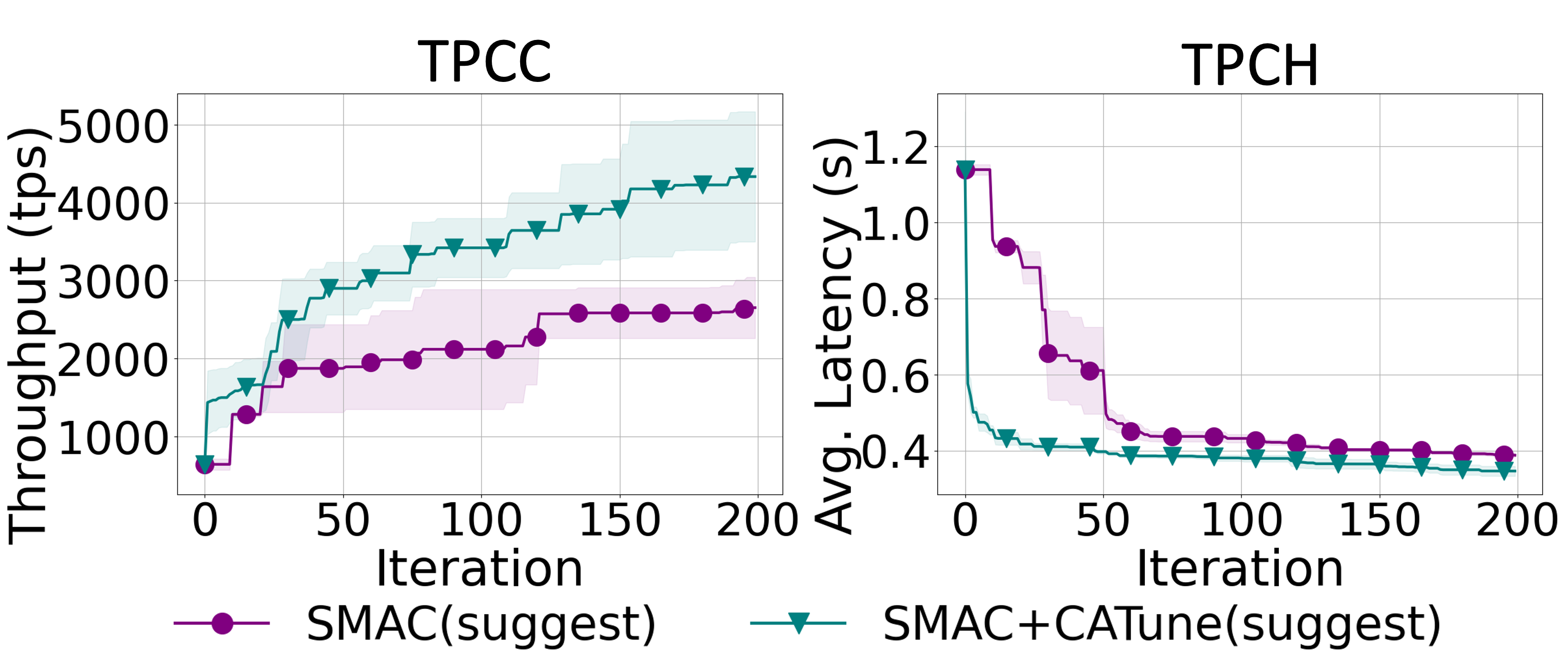}
        % \vspace{-2em}
        \caption{Best performance achieved by \Sys{} under knowledge-guided reduced ranges on PostgreSQL.}
        \label{fig:efficiency_performance_suggest}
    \end{minipage}
    % \vspace{-2em}
    
\end{figure*}

\subsection{Comparison with Manual-reading Systems}
\label{sec:extraction_comparison}

Our extraction objective differs from those of existing manual-reading tuning systems. Prior systems extract textual hints that recommend values for individual knobs and adapt or aggregate these recommendations using runtime feedback, or build per-knob structured knowledge such as suggested, minimum, maximum, and special values~\cite{10.1145/3514221.3517843,10.14778/3659437.3659449}. Such knowledge is primarily used for knob selection, range refinement, and tuning initialization.

In contrast, \Sys{} extracts \emph{inter-knob dependencies} rather than recommended values for individual knobs. Each extracted record explicitly represents the participating knob pair, relation direction, optional activation condition, supporting evidence, and confidence score. Supported pairwise ordering records are then translated into structural constraints on the configuration space.
Due to space limitations, we provide the detailed characterization of documented dependency types and the extraction quality evaluation in our technical report~\cite{CATuneTechReport}. The report also includes comparisons against alternative extraction strategies.

\section{Evaluation}
\label{sec:evaluation}

% We evaluate \Sys{} on two representative workloads to measure both tuning efficiency and final performance. Across different search spaces, \Sys{} reduces time-to-optimal by $2.27\times$ to $12.5\times$ compared to SMAC, while achieving performance gains ranging from 10.37\% to 63.37\%. The improvements remain consistent across different BO frameworks, indic;tating that structural constraint encoding is optimizer-agnostic. In particular, topology-aware sampling significantly accelerates convergence, achieving up to $12.5\times$ speedup over rejection-based strategy.

This section evaluates the impact of ordering constraints on tuning efficiency and  performance. We examine constraint-aware search across workloads, search-space regimes, and BO frameworks, and compare topology-aware sampling with rejection-based strategies.

\subsection{Experimental Setup}
\label{sec:exp_setup}

\textbf{Workloads.}
We evaluate our approach on two standard benchmarks covering analytical and transactional regimes: 
TPC-H (OLAP workload) with scale factor 1, and TPC-C (OLTP workload) with scale factor 200. 
TPC-C is configured with ten terminals and unlimited arrival rate to stress concurrency behavior. 
All benchmark implementations are from BenchBase \cite{10.14778/2732240.2732246}. 
These workloads are widely adopted in DBMS tuning studies \cite{10.14778/3659437.3659449}.

\noindent\textbf{Hardware.}
All experiments are conducted on a dedicated machine equipped with a 12-core 12th Gen Intel(R) Core(TM) i7-12700K CPU, 32\,GB RAM, and a 1\,TB Samsung SSD.

\noindent\textbf{Tuning Settings.}
We conduct all experiments using PostgreSQL v13 and MySQL v8.0. 
% We assume a candidate set of tunable knobs is provided, as in prior ML-based tuning systems that address knob selection separately 
% (e.g., GPTuner and DB-BERT) \cite{10.14778/3659437.3659449, 10.14778/3551793.3551844}. 
We assume a candidate set of tunable knobs is provided, since prior ML-based tuning systems treat knob selection as an independent preprocessing step (e.g., GPTuner and OpAdviser) \cite{10.14778/3659437.3659449, 10.14778/3632093.3632114}.
Our work focuses on discovering and enforcing structural dependencies within a given knob set, rather than on knob selection itself.
In our experiments,  we consider 45 and 50 performance-relevant knobs for PostgreSQL and MySQL, respectively,  spanning memory allocation, parallelism, and write-ahead logging behavior. 
We exclude parameters related to debugging, logging verbosity, security, and file path configuration, as these do not directly influence performance. 
From this knob set, we extract 10 documented ordering dependencies and enforce them in the constraint-aware setting. 
Each workload is tuned for 200 iterations. 
The first 10 configurations are generated via LHS (baseline) or topology-aware sampling (our method) as the initial design. 
We perform four independent runs with different random seeds to account for optimizer stochasticity and report the mean best performance across runs.
We optimize for overall system throughput (higher is better) for TPC-C and average latency (lower is better) for TPC-H, following common practice in prior tuning work \cite{10.14778/3659437.3659449}.

\noindent\textbf{Optimizers.}
We evaluate our method using two widely adopted Bayesian optimization frameworks: SMAC and Gaussian Process BO (GP-BO). 
Both are implemented within the SMAC3 infrastructure \cite{10.5555/3586589.3586643}, ensuring a unified experimental environment. 
For SMAC, we use its default Random Forest surrogate model with expected improvement. 
For GP-BO, we replace the surrogate with a Gaussian Process while retaining the same acquisition optimization pipeline. 
This unified setup isolates the effect of surrogate modeling from implementation-level artifacts.

To assess compatibility with knowledge-enhanced tuning, we additionally evaluate integration with GPTuner, a coarse-to-fine BO framework built on SMAC \cite{10.14778/3659437.3659449, 10.5555/3586589.3586643}.
Due to space limitations, the main paper presents detailed results on PostgreSQL; we conduct the same set of experiments on MySQL v8.0 and observe consistent improvements. Please refer to our technical report \cite{CATuneTechReport} for the complete MySQL results.

\subsection{Efficiency Evaluation}
\label{sec:efficiency_evaluation}

\noindent\textbf{Evaluation Dimensions.}
We evaluate \Sys{} along two dimensions:
(i) \textbf{final tuning performance} (best throughput or latency after 200 iterations), and 
(ii) \textbf{sample efficiency}, measured as the number of iterations required to reach the best observed performance (time-to-optimal).
Unless otherwise stated, SMAC with a random forest surrogate serves as the primary baseline.

\noindent\textbf{Scenarios.}
To isolate the effect of constraint-aware sampling from other forms of tuning knowledge, we consider four settings:

\begin{itemize}[nosep, left=2pt]
\item \textit{SMAC (default):} SMAC operating over the default knob ranges specified in the PostgreSQL manuals.

\item \textit{SMAC+\Sys{} (default):} SMAC augmented with constraint-aware tuning under the same default ranges.

\item \textit{SMAC (suggest):} SMAC using knowledge-guided suggested ranges derived from GPTuner’s knowledge handler~\cite{10.14778/3659437.3659449}.

\item \textit{SMAC+\Sys{} (suggest):} Suggested ranges combined with constraint-aware tuning.
\end{itemize}

This allows us to evaluate (i) the standalone impact of ordering constraints and 
(ii) their interaction with range-based knowledge.

\begin{figure*}[!t]
    \centering
    % \vspace{-2em}
    \begin{minipage}[t]{0.24\linewidth}
        \centering
        \includegraphics[width=\linewidth]{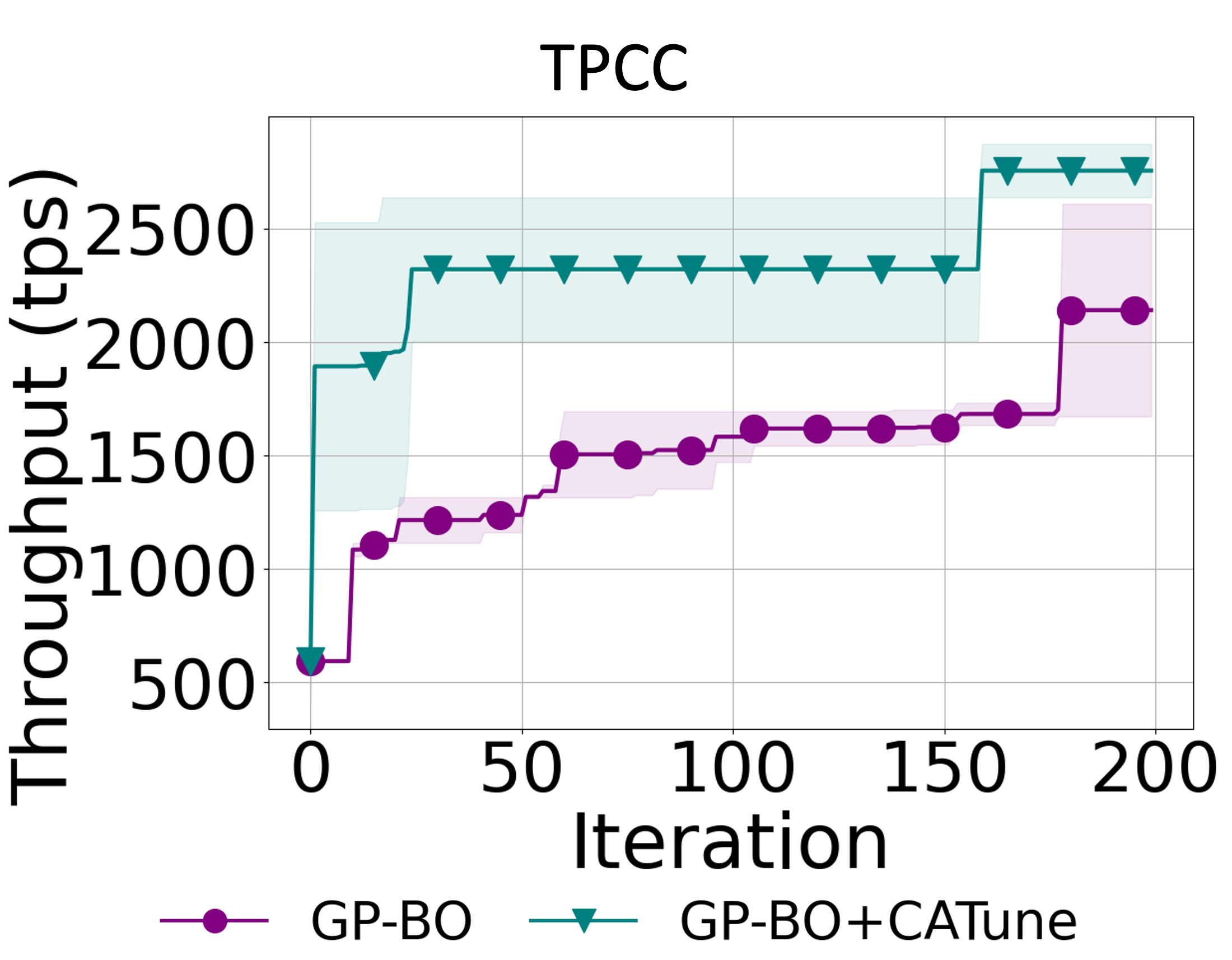}
        % \vspace{-2em}
        \caption{Performance of \Sys{} integrated with GP-BO under knowledge-guided suggested ranges.}
        \label{fig:GPBO_CATune_suggest}
    \end{minipage}
    \hfill
    \begin{minipage}[t]{0.24\linewidth}
        \centering
        \includegraphics[width=1\linewidth]{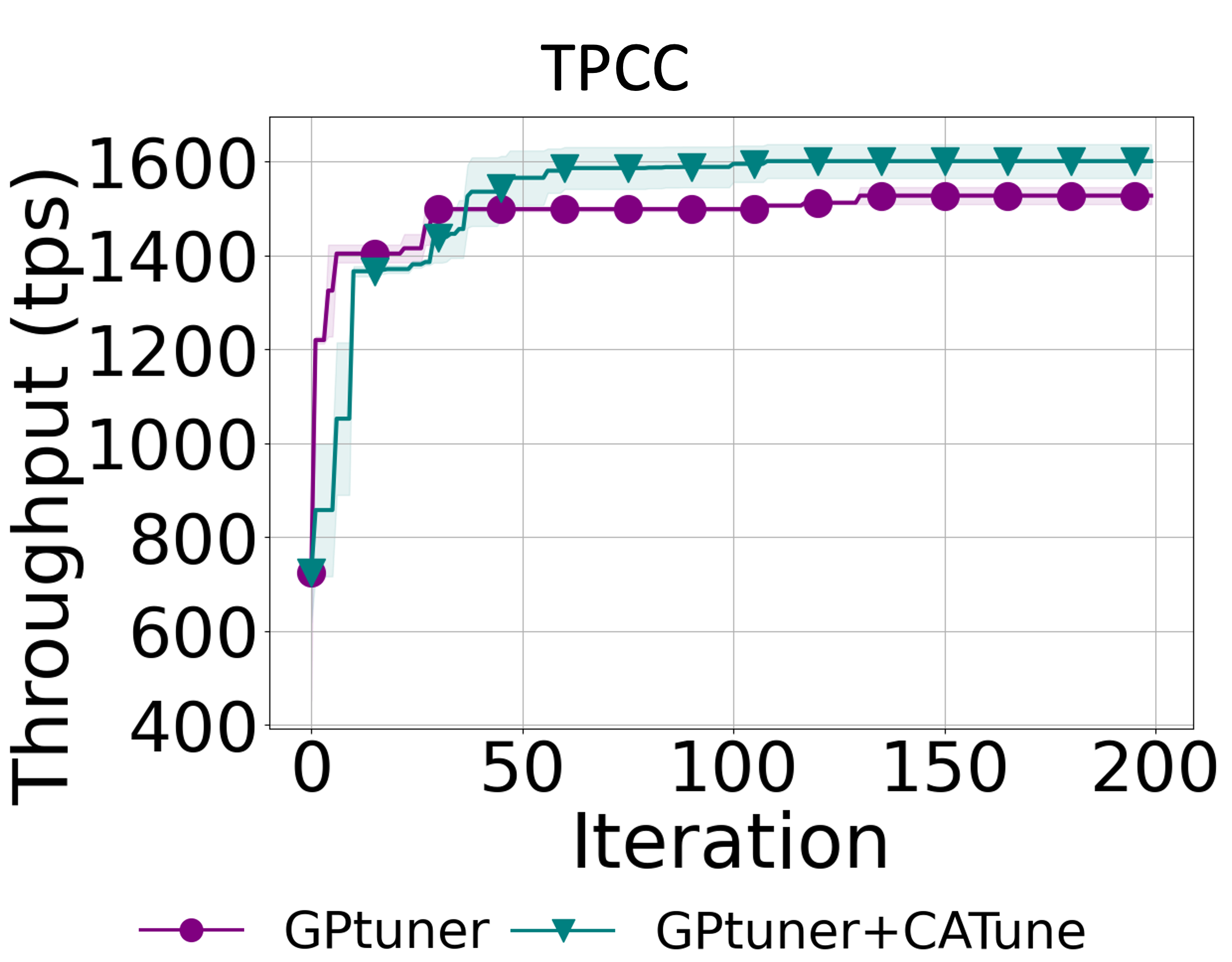}
        % \vspace{-2em}
        \caption{Performance impact of \Sys{} when integrated with coarse-to-fine BO framework from GPTuner.}
        \label{fig:GPTuner_CATune}
    \end{minipage}
    \hfill
    \begin{minipage}[t]{0.22\linewidth}
        \centering
        \includegraphics[width=\linewidth]{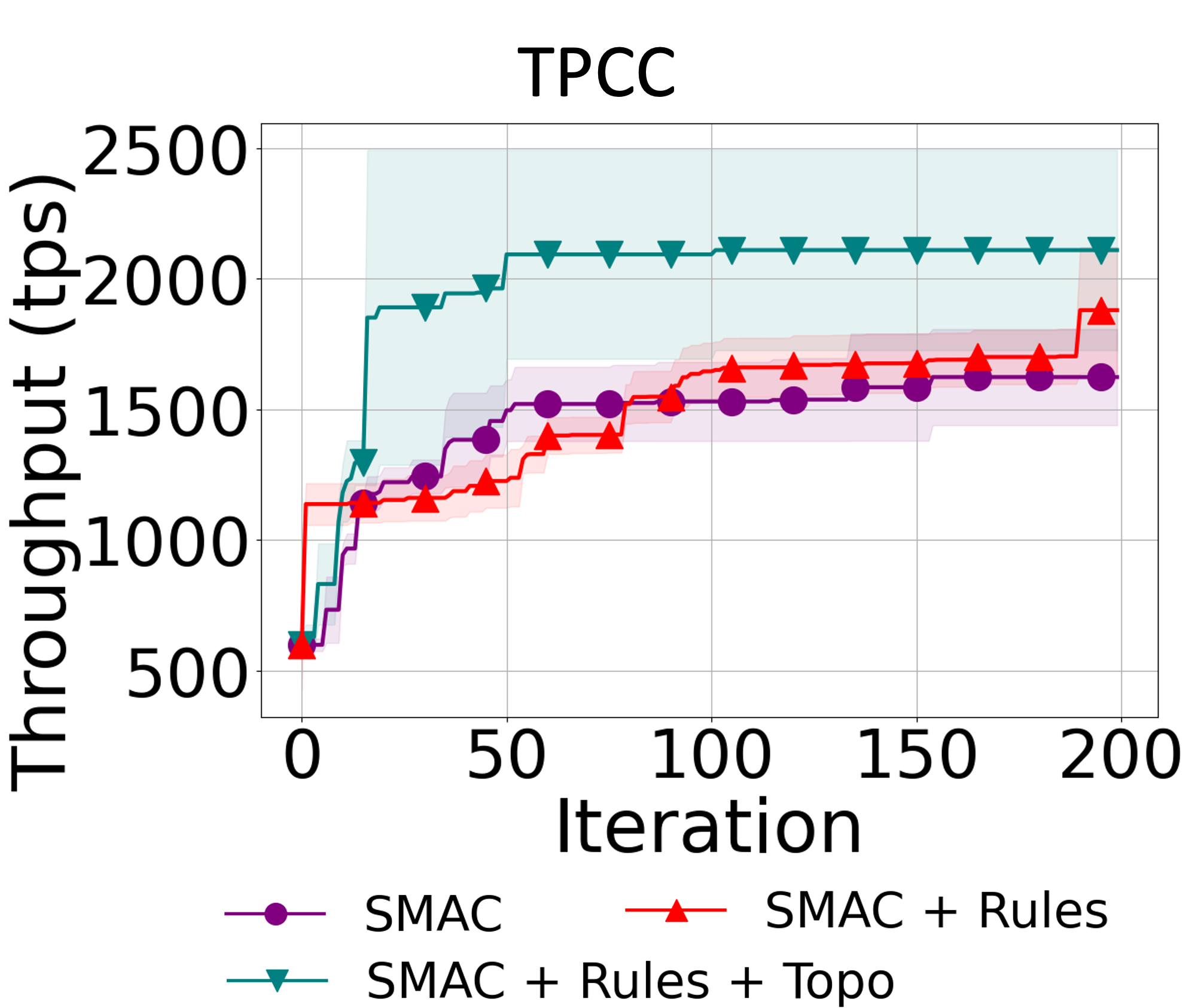}
        % \vspace{-2em}
        \caption{The result of ablation study for \Sys{}, presenting the contribution of topology-aware sampling.}
        \label{fig:rejection_sampling}
    \end{minipage}
    \hfill
    \begin{minipage}[t]{0.22\linewidth}
        \centering
        \includegraphics[width=1\linewidth]{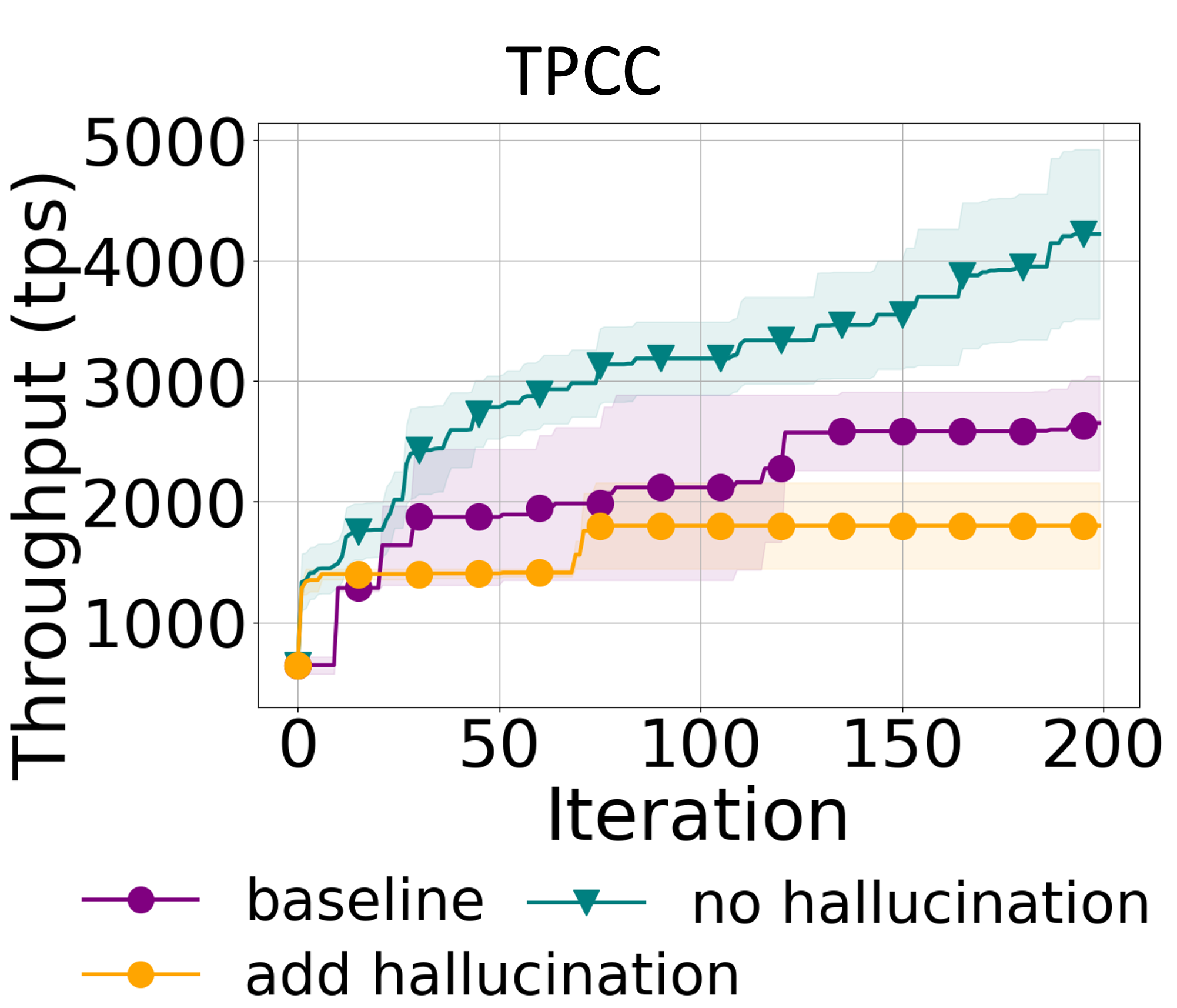}
        % \vspace{-2em}
        \caption{Performance impact between before and after adding constraints that reduce efficiency.}
        \label{fig:hallucination}
    \end{minipage}
    % \vspace{-1em}
\end{figure*}

\noindent\textbf{Results: Default Search Ranges.}
Figure~\ref{fig:efficiency_performance_default} reports results under vendor-provided default ranges. 
\Sys{} significantly improves sample efficiency, reaching the baseline SMAC optimum (after 200 iterations) $12.5\times$ faster on TPC-C and $2.27\times$ faster on TPC-H. 
It also yields superior final configurations: after 200 iterations, throughput improves by 33.39\% on TPC-C and latency decreases by 12.55\% on TPC-H relative to baseline SMAC.

\noindent\textbf{Results: Suggested Search Ranges.}
Figure~\ref{fig:efficiency_performance_suggest} shows results under knowledge-guided reduced ranges on PostgreSQL; the corresponding MySQL results are provided in the technical report \cite{CATuneTechReport}.
The gains remain consistent. \Sys{} reaches the baseline SMAC optimum $5.41\times$ faster on TPC-C and $3.57\times$ faster on TPC-H. 
In terms of final performance, it improves throughput by 63.37\% on TPC-C and reduces latency by 10.73\% on TPC-H compared to the baseline SMAC. On MySQL, \Sys{} achieves an 11.35\% higher throughput than baseline SMAC on TPC-C while reaching the baseline optimum $28.5\times$ faster \cite{CATuneTechReport}.

\noindent\textbf{Discussion.}
Comparing default and knowledge-guided ranges disentangles the effect of ordering constraints from simple range narrowing. 
Default ranges represent a cold-start regime with large infeasible or low-quality regions, whereas reduced ranges correspond to a warm-start regime where prior knowledge already prunes part of the domain. 
The consistent improvements in both settings indicate that constraint-aware tuning provides benefits beyond range restriction alone. 
Overall, \Sys{} improves both convergence speed and final tuning quality across deployment scenarios.

\subsection{Generalizing across Optimizers} 
\label{sec:general_optimizer_study}

\noindent\textbf{Scenarios.}
To demonstrate that \Sys{} is not tied to a specific optimizer or BO implementation, we evaluate it with different surrogate models and optimization frameworks under TPC-C workload. In addition to SMAC (random forest surrogate), we consider Gaussian Process-based BO (GP-BO) and the coarse-to-fine BO framework used in GPTuner. We follow the same experimental protocol as in Section~\ref{sec:efficiency_evaluation}, optimizing for final performance and measuring sample efficiency.
The evaluated scenarios are:

\begin{itemize}[nosep, left=2pt]
% \item \textbf{GPBO (default):} GP-based BO with manual default knob ranges.
% \item \textbf{GPBO+\Sys{} (default):} GP-based BO with constraint-aware \Sys{} under default ranges.
\item \textit{GPBO:} GP-based BO with knowledge-guided suggested ranges.
\item \textit{GPBO+\Sys{}:} GP-based BO with suggested ranges combined with \Sys{}.
\item \textit{GPTuner:} Coarse-to-fine BO with knowledge-guided ranges.
\item \textit{GPTuner+\Sys{}:} GPTuner augmented with \Sys{}.
\end{itemize}

\noindent\textbf{Results: GP-BO.}
% Figure~\ref{fig:GPBO_CATune_default} reports results under the manual default ranges. {\color{red}\Sys{} substantially improves sample efficiency: it reaches the baseline GP-BO optimum (measured after 200 iterations) $n\times$ faster on TPC-C and $1.63\times$ faster on TPC-H. In addition, \Sys{} improves final tuning quality. After 200 iterations, it achieves a n\% throughput improvement on TPC-C and a 9.34\% latency reduction on TPC-H relative to baseline GP-BO. }
% Second, \Sys{} also improves final tuning quality. After 200 iterations, \Sys{} achieves an average throughput improvement of n\% on TPC-C and a latency improvement of n\% on TPC-H over baseline GP-BO.
Figure~\ref{fig:GPBO_CATune_suggest} shows results under knowledge-guided reduced ranges. The improvements remain consistent. \Sys{} reaches the baseline GP-BO optimum $8.33\times$ faster on TPC-C. In terms of final performance, \Sys{} improves throughput by 28.71\% on TPC-C compared to baseline GP-BO.

\noindent\textbf{Results: GPTuner.}
We next evaluate \Sys{} when coupled with GPTuner’s coarse-to-fine BO framework. The same trend holds: \Sys{} accelerates convergence and improves final performance. \Sys{} reaches the baseline GPTuner optimum $5.26\times$ faster and improves throughput by 4.81\% on TPC-C as shown in Figure \ref{fig:GPTuner_CATune}. 

\noindent\textbf{Summary.}
Across different surrogate models (random forest vs. Gaussian process) and BO frameworks (standard GP-BO vs. coarse-to-fine GPTuner), \Sys{} consistently improves both sample efficiency and final performance. These results indicate that constraint-aware search-space restriction is orthogonal to the underlying optimizer and generalizes across diverse BO implementations.

\subsection{Scalability Study}
\label{sec:scalability_study}
We conduct a scalability study to examine whether the benefits of \Sys{} persist as the database size increases. Due to space limitations, we omit the study from the main paper and refer readers to the technical report~\cite{CATuneTechReport} for detailed results and analysis.

\subsection{Ablation Study}
\label{sec:ablation_study_topo}

\noindent\textbf{Scenarios.}
To isolate the contribution of topology-aware sampling in \Sys{}, we compare the following settings:

\begin{itemize}[nosep,left=2pt]
\item \textit{SMAC (No rules, baseline):} The original SMAC optimizer operating over the unconstrained configuration space without explicit constraint handling.

\item \textit{SMAC + Rules (Rejection):} A hard-constraint enforcement baseline. Configurations that violate extracted dependency constraints are rejected and resampled until a valid configuration is obtained. This represents a conventional post-hoc constraint-handling strategy.

\item \textit{SMAC + Rules + Topo (\Sys{}):} The full \Sys{} design. Topology-aware sampling uses the dependency graph to generate configurations that satisfy ordering constraints during sampling, thereby avoiding infeasible regions by construction.
\end{itemize}

\noindent\textbf{Results.}
Figure~\ref{fig:rejection_sampling} shows that rule-based constraint enforcement improves over unconstrained SMAC by avoiding evaluations of invalid configurations. Topology-aware sampling further improves both convergence and final performance by generating feasible configurations directly instead of repeatedly rejecting infeasible ones. Compared with SMAC + Rules, the full system achieves 18.8\% higher throughput and reaches the same performance level 2.19$\times$ faster, demonstrating the benefit of incorporating structural dependencies during sampling.

\noindent\textbf{Summary.}
These results show that hard constraint enforcement alone is insufficient. While rule-based validation improves tuning quality, topology-aware sampling further increases sample efficiency by generating feasible configurations by construction, making it a key component of \Sys{}.

\section{Discussion}

% \subsection{Operational Impact of Constraints}
% As Section \ref{sec:reliability_risk} elaborated, the violations of different ordering constraints have different impacts: DBMS startup failures or reducing efficiency. Here, we study the impact of these two types of constraints. 
% \paragraph{Impact of DBMS startp failures} We use the ordering constraint, \texttt{superuser\_reserved\_connections} $<$ \texttt{max\_connections}, as an example to study. We conduct experiments with 10 knobs including \texttt{superuser\_reserved\_connections} and \texttt{max\_connections} with LLM-suggested ranges. Throughout the 200 iterations, there are 9 configurations that failed because of DBMS startup failure, in which, all failure configurations violate the ordering constraint between \texttt{superuser\_reserved\_connections} and \texttt{max\_connections}.

% \paragraph{Impact of constraints that reduce efficiency} Among the constraint sets used in Section \ref{sec:evaluation}, we add two additional fake constraints that unnecessarily reduce search spaces. Figure \ref{fig:study_constraint_quality} show a improvement downgradement 55.75\%.

This section analyzes how ordering constraints influence system stability and high-performing configurations, and we compare proactive filtering against penalty-based constraint handling.

\subsection{Operational Impact of Constraints}
\label{sec:operational_impact}

Ordering constraints influence tuning in two distinct ways. 
First, they prevent configurations that would cause system failures. 
Second, they reshape the feasible region in ways that may either improve or hinder optimization. 
We examine both effects empirically.

We begin with correctness-critical constraints whose violation prevents DBMS startup. 
As a representative example, PostgreSQL requires 
\texttt{superuser\_reserved\_connections} $<$ \texttt{max\_connections}. 
We conduct tuning over 10 knobs, including these parameters, using LLM-suggested ranges. 
Across 200 iterations, 9 sampled configurations failed to start the DBMS, and all failures corresponded to violations of this ordering constraint. 
This result demonstrates that correctness-critical ordering constraints are essential for reliability: without proactive enforcement, the optimizer wastes evaluations on configurations that cannot even initialize the system.
We next examine the opposite scenario—constraints that unnecessarily restrict the search space. 
To study this effect, we inject two hallucinated ordering constraints produced by a naive LLM-based extraction method. 
We compare three settings: baseline SMAC with suggested ranges, \Sys{} with valid constraints, and \Sys{} augmented with the two hallucinated constraints. 
Figure~\ref{fig:hallucination} shows that hallucinated constraints severely degrade tuning performance. 
Without hallucinations, \Sys{} reaches the baseline optimum $4.55\times$ faster and achieves a 59.15\% improvement in final performance. 
When the two incorrect constraints are injected, the achievable performance improvement drops by 32.07\% relative to the baseline. 
Compared to the hallucinated setting, \Sys{} without hallucinations achieves 134.29\% higher final performance and converges $9.09\times$ faster.

These findings highlight a fundamental trade-off. 
Injecting semantically valid constraints improves reliability and sample efficiency by removing invalid regions while preserving high-performing configurations. 
In contrast, hallucinated or overly restrictive constraints artificially truncate the feasible domain and may exclude near-optimal configurations. 
This underscores the importance of high-precision constraint extraction and validation when integrating automatically generated constraints into constraint-aware tuning systems.

\subsection{Efficiency of Different Sampling Strategies}
\label{sec:efficiency_topo_sampling}

To isolate the sampling efficiency of topology-aware sampling, we compare it with three alternative constraint-handling strategies on PostgreSQL under the default knob ranges and tuning settings described in Section~\ref{sec:exp_setup}:
\begin{itemize}[nosep, left=2pt]
\item \textit{No Rules}: configurations are sampled from the original search space without constraint handling.
\item \textit{Reject}: configurations are repeatedly sampled from the original search space until they satisfy all constraints.
\item \textit{Repair}: configurations are first sampled from the original search space and then modified to satisfy violated constraints.
\item \textit{Topology-aware Sampling} (Topo): constraints are incorporated before sampling, so generated configurations are feasible by construction.
\end{itemize}

\begin{table}[t]
\centering
% \caption{Sampling efficiency under PostgreSQL in tuning settings with default knob ranges over 100 requested configurations (Returned). Generated indicates the actual number of generated configurations. Invalid means the number of configurations that violate at least one rule. Time (s) is the timing of returning 100 valid configurations. *indicates the number of repaired configurations}
\caption{Sampling efficiency for PostgreSQL tuning over 100 requested samples (Returned). Generated is the number of configurations generated, Invalid the number violating $\geq1$ rule, and Time (s) the time to obtain 100 valid configurations. $^\ast$ denotes repaired configurations.}
% \vspace{-1em}
\label{tab:sampling-efficiency}
\begin{tabular}{lrrrr}
\toprule
\textbf{Method} & \textbf{Returned} & \textbf{Generated} & \textbf{Invalid} & \textbf{Time (s)} \\
\midrule
No Rules & 100 & 100 & 99 & 0.026 \\ 
Reject & 100 & 15,992 & 15,892 & 6.841 \\
Repair & 100 & 100 & 99* & 0.031 \\
Topo & 100 & 100 & 0 & 0.089 \\
\bottomrule
\end{tabular}
% \vspace{-2em}
\end{table}

\noindent\textbf{Results.}
Table~\ref{tab:sampling-efficiency} shows that unconstrained sampling is highly inefficient: 99 out of 100 configurations violate at least one constraint. Rejection sampling avoids invalid outputs, but it must generate 15,992 candidates to obtain 100 feasible ones, rejecting 15,892 samples and taking 6.841 seconds. Repair also starts from the unconstrained space; since 99 out of 100 samples require repair, it introduces an extra post-processing step before configurations can be used. In contrast, topology-aware sampling generates feasible configurations directly, returning 100 valid samples from exactly 100 generated candidates without rejection or repair. This avoids both wasted trials and post-hoc modification, making topology-aware sampling more efficient for constrained tuning.

\subsection{Impact of Different Constraint Employment Strategies}
\label{sec:different_strategies}

\Sys{} enforces ordering constraints by proactively restricting exploration and exploitation to the feasible region. 
An alternative approach is to retain the full search space and penalize constraint violations during optimization. 
We compare these two strategies:

\begin{itemize}[nosep, left=2pt]
\item \textit{Proactive Filtering (\Sys):} Infeasible regions are removed from the search space before sampling.
\item \textit{Passive Penalization (MultiObj):} The optimizer searches the full domain while minimizing performance loss and constraint violations simultaneously.
\end{itemize}

For the passive strategy, we formulate tuning as a multi-objective problem. 
For an ordering constraint $\theta_i \le \theta_j$, the violation penalty is defined as $\mathcal{P}_{ij}(\theta)
=
\frac{\max(0,\, \theta_i - \theta_j)}
{|\theta_j| + 1}$.
The total penalty across all ordering constraints is $\mathcal{P}(\theta)
=
\sum_{(i,op,j)\in \mathcal{C}}
\mathcal{P}_{ij}(\theta)$,
% \begin{equation}
% \mathcal{P}(\theta)
% =
% \sum_{(i,op,j)\in \mathcal{C}}
% \frac{\max(0,\, \theta_i - \theta_j)}
% {|\theta_j| + 1},
% \end{equation}
where $op \in \{<,\le\}$. The denominator normalizes violations to reduce scale sensitivity across knobs with heterogeneous ranges.
We implement this strategy using ParEGo~\cite{1583627} within the SMAC framework.

\begin{example}
Consider an ordering constraint \texttt{superuser \_reserved\_connections} $\le$ \texttt{max\_connections} with 
\texttt{superuser\_ reserved\_connections} $= 120$ and \texttt{max\_connections} $= 100$. 
The violation magnitude is $120 - 100 = 20$, and the normalized penalty becomes
$
\mathcal{P}_{ij}(\theta)
=
\frac{20}{|100| + 1}
\approx 0.198.$
If \texttt{superuser\_reserved\_connections}$=4$, \texttt{max\_connections}$=6$, the penalty is zero.
\end{example}

\noindent\textbf{Results.}
Figure~\ref{fig:study_different_strategies} shows the convergence curves on TPC-C and TPC-H. On average, \Sys{} reaches the baseline optimum approximately $\sim 3.93\times$ and $3.57\times$ faster than the baseline on TPC-C and TPC-H, respectively, while MultiObj achieves only $2.3\times$ and $1.74\times$ speedups. This indicates that proactively restricting the search space leads to substantially faster convergence.
After 200 iterations, \Sys{} improves average throughput on TPC-C by 53.52\% over the baseline, compared to 38.26\% achieved by MultiObj. On TPC-H, \Sys{} reduces latency by 10.73\%, whereas MultiObj achieves a 3.43\% reduction. Overall, \Sys{} consistently delivers larger performance improvements across both workloads.

\noindent\textbf{Discussion.}
The weaker performance of passive penalization arises from its need to explore infeasible configurations within the full search space. 
Although violations are penalized, evaluations are still consumed by invalid or low-quality regions. 
In contrast, proactive filtering removes infeasible regions entirely, concentrating the optimization process on semantically valid configurations.

These results demonstrate that structural constraint injection is more sample-efficient and effective than penalty-based handling.

\begin{figure}
    \centering
    \includegraphics[width=\linewidth]{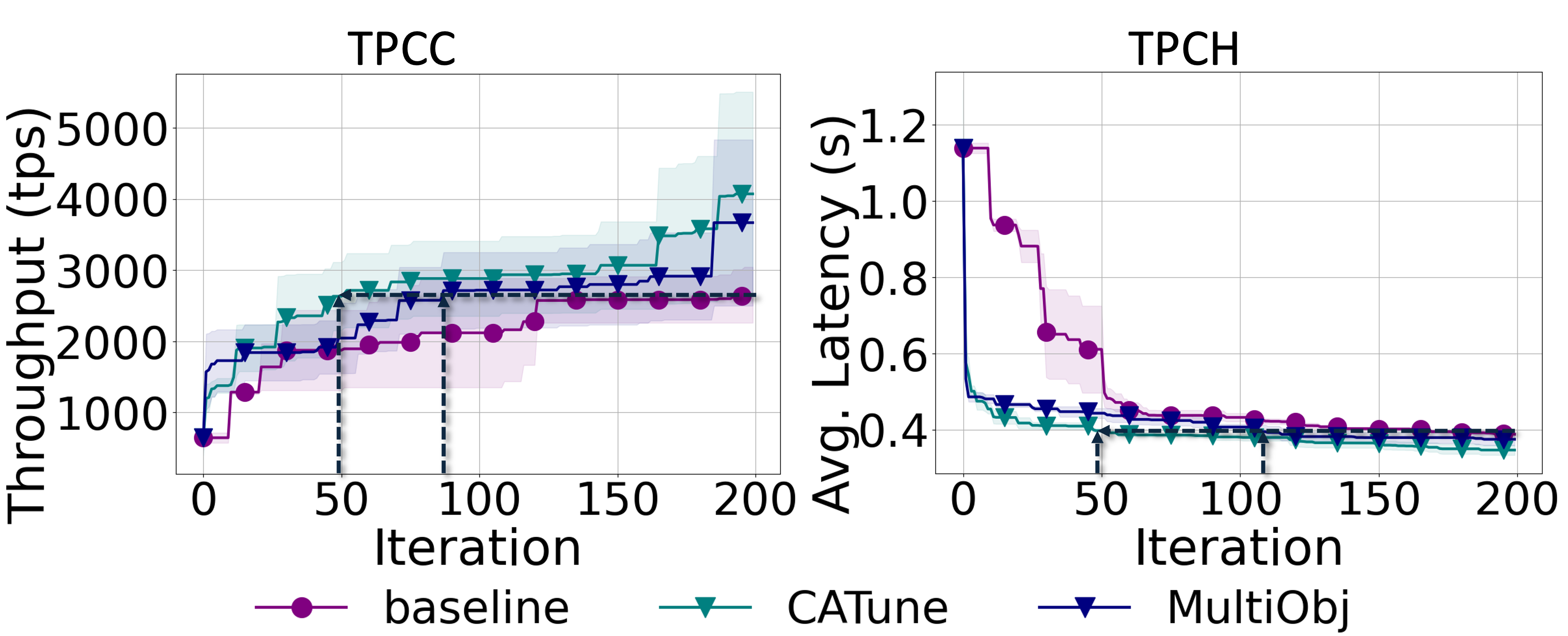}
    % \vspace{-2em}
    \caption{Effect of different employment strategies: proactive filtering (\Sys) vs. passive penalization (MultiObj).}
    \label{fig:study_different_strategies}
    % \vspace{-2.5em}
\end{figure}

% \begin{figure}[]
%     \centering
%     \begin{minipage}[t]{0.48\linewidth}
%         \centering
%         \includegraphics[width=\linewidth]{figures/ablation.png}
%         \vspace{-2em}
%         \caption{The result of ablation study for \Sys{}, presenting the contribution of topology-aware sampling.}
%         \label{fig:rejection_sampling}
%     \end{minipage}
%     \hfill
%     \begin{minipage}[t]{0.48\linewidth}
%         \centering
%         \includegraphics[width=1\linewidth]{figures/hallucination.png}
%         \vspace{-2em}
%         \caption{Performance impact between before and after adding constraints that reduce efficiency.}
%         \label{fig:hallucination}
%     \end{minipage}
%     \vspace{-1.5em}
% \end{figure}

\section{Related Work}
This section reviews prior work on database configuration tuning and constrained optimization.

\noindent\textbf{Database Configuration Tuning.}
Automated database configuration tuning has been studied extensively and can be broadly categorized into heuristic-based, Bayesian Optimization (BO)-based \cite{NIPS2012_05311655}, and Reinforcement Learning (RL)-based approaches \cite{10.5555/3504035.3504427}.
Rule-based methods rely on manually crafted tuning rules to guide exploration \cite{10.5555/1287369.1287454, Kwan2002AutomaticCF}. 
Search-based methods employ heuristic search strategies (e.g., avoiding revisits and exploring local neighborhoods) to iteratively refine configurations \cite{10.1145/2628071.2628092, 10.1145/3127479.3128605}. 
These approaches become less effective as the configuration space grows in dimensionality and exhibits complex inter-knob dependencies.
BO-based methods \cite{10.14778/3457390.3457404, 10.14778/1687627.1687767, 10.1145/3318464.3380591, 10.1145/3035918.3064029, 10.1145/3448016.3457291, 10.1145/3514221.3526176} treat tuning as a black-box optimization problem. 
They iteratively fit a probabilistic surrogate model that maps knob configurations to performance and select new configurations by maximizing an acquisition function. 
These methods improve sample efficiency but typically assume box-constrained search spaces and do not explicitly model structural dependencies among knobs.
RL-based approaches explore the configuration space through sequential decision-making, balancing exploration and exploitation via agent--environment interaction \cite{10.1145/3514221.3517882, 10.14778/3352063.3352129, 10.1145/3464389, 10.14778/3484224.3484236, 10.1145/3299869.3300085}. 
While effective in dynamic settings, RL methods similarly operate over nominal configuration domains without enforcing structural validity constraints by construction.

\noindent\textbf{Constrained Bayesian Optimization.}
Constrained Bayesian Optimization (CBO) extends BO to handle unknown or partially observed constraints \cite{Schonlau1998, 10.5555/3020751.3020778, 10.5555/3044805.3044997, 10.1007/s11222-014-9477-x}. 
In classical CBO, constraints are modeled as black-box functions with probabilistic surrogates, and feasibility is incorporated through acquisition strategies such as Expected Improvement with Constraints or Probability of Feasibility. 
These methods generally require sampling infeasible configurations to learn constraint boundaries.
In database tuning, ResTune \cite{10.1145/3448016.3457291} applies constrained optimization to enforce service-level objectives, such as throughput or latency bounds, which are workload-dependent and evaluated only after executing candidate configurations. 
Feasibility is therefore determined post hoc through runtime measurements.

Our setting differs fundamentally. 
We focus on deterministic knob dependency constraints that are intrinsic to the configuration space (e.g., ordering constraints among knobs). 
Such constraints are static, system-defined, and independent of workload behavior. 
Instead of modeling feasibility probabilistically or verifying it after execution, we encode structural dependencies directly into the configuration domain so that all sampled configurations are valid by construction. 
This eliminates infeasible trials \emph{a priori}, which is particularly important in database systems where dependency violations may cause instability or startup failures \cite{10.14778/3476311.3476411, 10.14778/3450980.3450992}.

From an optimization perspective, our approach restricts the search to a constraint-consistent subspace by encoding structural constraints directly into the feasible region.
This differs from GPTuner \cite{10.1145/3733620.3733641}, which leverages domain knowledge to refine knob ranges and prioritize important parameters. 
While GPTuner biases exploration toward promising regions, it does not enforce structural validity among knobs. 
Our method instead guarantees feasibility at the configuration-space level. 
In contrast to prior database tuning systems that treat the configuration space as box-constrained, our approach explicitly encodes deterministic ordering constraints as structural components of the search domain.

\section{Limitations}
\label{sec:limitations}

\noindent\textbf{Supported constraint types.}
Although our extraction pipeline identifies diverse inter-knob constraints, \Sys{} currently utilizes only ordering and resource-bounded constraints. Supporting richer constraint types and effectively incorporating them into optimization remains future work.

\noindent\textbf{Coverage of supported constraints.}
The prevalence of ordering constraints varies across DBMSs (e.g., only 2.1\% in ClickHouse). We position \Sys{} as complementary to existing tuners: when fewer supported constraints exist, it imposes fewer additional restrictions, while applicable constraints can still eliminate invalid or undesirable configurations. Nevertheless, broader constraint coverage could further improve its applicability and effectiveness.

\noindent\textbf{Integration with search-space reduction methods.}
\Sys{} is generally complementary to existing techniques such as low-dimensional projections~\cite{10.14778/3551793.3551844}, suggested values~\cite{10.1145/3514221.3517843,10.1145/3514221.3520171}, suggested ranges, and special-value handling~\cite{10.14778/3659437.3659449,10.1145/3733620.3733641,10.1145/3626246.3654739}. However, integration with low-dimensional projections requires specialized designs because their many-to-one knob mappings may not preserve constraints defined over the original knobs. Developing constraint-preserving projections remains future work.
\section{Conclusion}

We proposed \Sys{}, a constraint-aware DBMS tuning framework that incorporates deterministic knob dependencies directly into Bayesian optimization. By eliminating infeasible configurations by construction, \Sys{} improves tuning reliability and efficiency. Our results demonstrate the value of explicitly modeling structural dependencies in DBMS configuration spaces.

% This work reframes DBMS tuning as a problem of structural search-space design. Rather than treating the configuration domain as box-constrained and discovering invalid regions through costly trial-and-error, we demonstrate that deterministic knob dependencies define a structurally consistent subspace that can be enforced prior to optimization. Hence, we proposed \Sys{}, a constraint-aware tuning framework that embeds ordering constraints directly into Bayesian optimization. By enforcing structural consistency during both exploration and exploitation, and by introducing topology-aware sampling, \Sys{} eliminates infeasible regions by construction and avoids the inefficiencies of rejection-based strategies.

% More broadly, our results show that explicitly modeling intrinsic configuration dependencies improves both reliability and convergence efficiency. Incorporating structural knowledge into the search domain provides a principled alternative to purely black-box tuning and opens new opportunities for constraint-aware optimization in complex system configuration tasks.
% In future work, we plan to extend this framework to support richer classes of structural constraints and to further investigate their interaction with advanced optimization strategies.
\begin{acks}
This work was supported in part by the U.S. National Science Foundation awards III-2107213 and ITE-2333789.
\end{acks}

%\clearpage

\bibliographystyle{ACM-Reference-Format}
\bibliography{bo, dbtuning, others}

\end{document}